\documentclass[11pt,letterpaper]{article}

\usepackage{fullpage}

\usepackage[utf8]{inputenc}
\usepackage[english]{babel}
\usepackage{amsmath}
\usepackage{xspace}
\usepackage{amsthm}
\usepackage{amssymb}
\usepackage{thmtools,complexity}
\usepackage{mathtools,thm-restate}
\usepackage{tabularx}
\usepackage{todonotes}
\usepackage{multirow}
\usepackage{comment}
\usepackage{soul}
\usepackage{lineno}
\usepackage{fancyhdr}

\usepackage[hidelinks]{hyperref}
\usepackage[capitalise]{cleveref}

\usepackage[T1]{fontenc}
\usepackage{tikz}
\usetikzlibrary{arrows,calc,positioning,arrows.meta,backgrounds,fit,decorations.pathmorphing,decorations.pathreplacing, shapes.geometric, calc}

\tikzstyle{intt}=[draw,text centered,minimum size=6em,text width=5.25cm,text height=0.34cm]
\tikzstyle{intl}=[draw,text centered,minimum size=2em,text width=2.75cm,text height=0.34cm]
\tikzstyle{int}=[draw,minimum size=2.5em,text centered,text width=3.5cm]
\tikzstyle{intg}=[draw,minimum size=3em,text centered,text width=6.cm]
\tikzstyle{sum}=[draw,shape=circle,inner sep=2pt,text centered,node distance=3.5cm]
\tikzstyle{summ}=[drawshape=circle,inner sep=4pt,text centered,node distance=3.cm]

\makeatletter

\newcommand{\labeltext}[3][]{%
    \@bsphack%
    \csname phantomsection\endcsname
    \def\tst{#1}%
    \def\labelmarkup{}
    \def\refmarkup{}%
    \ifx\tst\empty\def\@currentlabel{\refmarkup{#2}}{\label{#3}}%
    \else\def\@currentlabel{\refmarkup{#1}}{\label{#3}}\fi%
    \@esphack%
}
\makeatother

\theoremstyle{plain}
\newtheorem{theorem}{Theorem}[section]

\newtheorem{lemma}[theorem]{Lemma}
\newtheorem{observation}[theorem]{Observation}
\newtheorem{cor}[theorem]{Corollary}

\newtheorem{definition}[theorem]{Definition}

\theoremstyle{remark}

\newtheorem{claim}[theorem]{Claim}
\newenvironment{claimproof}[1][\proofname{}]{%
  \begin{proof}[#1]%
}{%
  \end{proof}%
}

\usepackage[numbers]{natbib}

\usepackage{enumitem}
\setlist[enumerate]{label*=(\arabic*),
	topsep=0.25em,
	parsep=0.1em,
	itemsep=0.1em
}
\setlist[itemize]{
	topsep=0.25em,
	parsep=0.1em,
	itemsep=0.1em
}

\usepackage[most]{tcolorbox}
 \title{Solving SAT via Treewidth Backdoors in Linear Time:  \\ 
 Beyond Sparsity}

\author{%
  Jan Dreier\thanks{Hasso Plattner Institute, Potsdam, Germany},
  Robert Ganian\thanks{Algorithms and Complexity Group, TU Wien, Austria},
  Marlene Gr\"undel\footnotemark[2],
  Roohani Sharma\thanks{Discrete Mathematics Group, Institute for Basic Science (IBS), Daejeon, South Korea},
  Simon Wietheger\footnotemark[2]%
}

\date{}

\newcommand{\backdoorextproblemtw}{\textsc{Extension Weak Backdoor to Bounded Treewidth}}

\newcommand{\bdextproblemtw}{\textsc{Ext-WB-TW}\xspace}

\newcommand{\bdextproblemtwbare}{\textsc{Ext-WB-TW}\xspace}

\renewcommand{\AA}{\mathcal{A}}

\renewcommand{\EE}{\mathcal{E}}

\renewcommand{\PP}{\mathcal{P}}
\newcommand{\QQ}{\mathcal{Q}}
\newcommand{\RR}{\mathcal{R}}

\DeclareMathOperator{\tw}{tw}

\newcommand{\paramproblemdef}[4]{
    \begin{tcolorbox}[
            enhanced,
            colback=white,
            colframe=black,
            boxrule = 0.5pt,
            coltitle=black,
            title=#1,
            rounded corners,
            bottom=-0.5mm,
            attach boxed title to top left={yshift=-10pt, xshift=6pt},
            boxed title style={
                    interior style={fill=white},
                    frame hidden
                }
        ]
        \begin{tabularx}{15cm}{ r X p{0.5cm} }
            Input:     & #2        \\
            Parameter: & #3        \\
            Question:  & #4
        \end{tabularx}
    \end{tcolorbox}
}

\newcommand{\var}{\textnormal{\texttt{Var}}}
\newcommand{\clause}{\textnormal{\texttt{Clauses}}}

\definecolor{ultramarine}{rgb}{0.07, 0.04, 0.7}

\begin{document}

\maketitle

\begin{abstract}

    The Boolean Satisfiability (\textsc{SAT}) problem is a remarkably expressive computational problem, serving both as a cornerstone of complexity theory and as a unifying language for a wide range of search and reasoning tasks. 
    
    While \NP-complete in general, 
    \textsc{SAT} is well known to be linear-time solvable on instances with bounded treewidth. This motivates the search for ``hidden'' tree-like structures in formulas---for instance by detecting so-called \emph{weak $t$-treewidth backdoors}. Such backdoors are sets of $k$ variables that can be assigned in such a way that the remaining instance has treewidth at most $t$ and is satisfiable whenever the original instance is. 
    
    Ever since Gaspers and Szeider [ICALP 2012] showed that finding a weak $t$-treewidth backdoor for a \textsc{SAT} formula is $\W[2]$-hard with respect to $k$,
    tractability results have been restricted to structurally {\em sparse} formulas---most notably,
    a linear-time algorithm for $d$-\textsc{SAT} (bounded clause-length SAT) by Fomin, Lokshtanov, Misra, Ramanujan, and Saurabh [SODA 2015]
    and a non-uniform polynomial-time algorithm when the incidence graph of the input formula belongs to a nowhere-dense graph class (or more generally is $K_{d,d}$-free)
    by Lokshtanov, Panolan, and Ramanujan [ICALP 2022].

    In this paper, we push the tractability boundaries of finding weak $t$-treewidth backdoors {\em beyond sparsity} and into the regime of ``well-behaved'' dense formulas. 
    Aiming towards relaxing structural input restrictions
    from nowhere-dense to {\em edge-stable} graph classes, as featured in the theory of beyond sparsity for FO model checking,
    or to {\em bounded VC-dimension} set systems,
    we identify {\em signed semi-ladder-freeness} as the correct dense structural restriction that establishes the boundary between tractability and intractability along this line.

    Algorithmically, we concentrate on formulas whose signed incidence graph is $d$-\emph{semi-ladder-free}, which include not only all $K_{d,d}$-free graphs, but also dense graph classes such as map graphs and more generally all graph classes that can be encoded by existential positive first-order formulas in nowhere-dense classes. 
    As our main technical contribution, we provide a linear-time algorithm that, for all fixed integers $d$, $k$, and $t$, takes as input a signed $d$-semi-ladder-free \textsc{SAT} formula $\varphi$ and either outputs a satisfying assignment for $\varphi$ along with a weak $t$-treewidth backdoor of size at most $k$, or detects that no such backdoor exists. 

    Complementing the above positive result, we provide \W[2]- and \W[1]-hardness results when the input is restricted to unsigned $d$-semi-ladder-freeness and signed $d$-ladder-freeness (edge-stable), respectively.
    This identifies signed $d$-semi-ladder-freeness as the natural tractability frontier for detecting weak $t$-treewidth backdoor sets in dense \textsc{SAT} instances, when trying to approach edge-stable graph classes or bounded VC-dimension set systems.

\end{abstract}

\section{Introduction}
\label{sec:intro}

The Boolean Satisfiability (\textsc{SAT}) problem is undeniably one of the most fundamental problems in computer science as it lies at the heart of the classical theory of $\NP$-completeness as well as of the Exponential Time Hypothesis and its variants. 
Given the theoretical hardness of \textsc{SAT}, the surprising efficiency with which real-world solvers can exploit hidden structural properties to solve \textsc{SAT} instances~\cite{sathandbook} has motivated a long line of research aimed at understanding the more fine-grained boundaries of tractability for the problem. Nowadays, this is often pursued via the framework of parameterized complexity~\cite{CyganFKLMPPS15}. 
A classical result in this domain is that \textsc{SAT} is fixed-parameter tractable (FPT) with respect to the treewidth of the instance, or more specifically, the treewidth of its incidence graph~\cite{SamerS10}\footnote{The incidence graph of a formula $\varphi$ contains a vertex for each variable and clause of $\varphi$, with edges connecting clauses with variables whose (positive or negative) literal occurs in that clause.}. 

Naturally, a \textsc{SAT} instance need not have bounded treewidth---yet it may still possess a ``hidden'' tree-like structure that reveals itself once a small number of variables is assigned. 
This phenomenon is typically viewed through the lens of \emph{backdoors}, introduced by Williams, Gomes, and Selman~\cite{WilliamsGS03}. 
A set $X$ of $k$ variables is a \emph{strong $t$-treewidth backdoor}
if every assignment of $X$ results in an instance of treewidth at most $t$; 
this means that for any fixed $k$ and $t$, we can solve \textsc{SAT} in linear time when provided with such a strong $t$-treewidth backdoor. This motivates the study of computing such backdoors of size at most~$k$.

Due to a classical result of Gaspers and Szeider~\cite{GaspersS13}, it is known that \textsc{SAT} is fixed-parameter tractable w.r.t.\ the size $k$ 
of a strong $t$-treewidth backdoor---in particular, for every fixed~$k$ and $t$, we can in cubic time either compute a strong $t$-treewidth backdoor of size bounded by a function of $k$, or determine that no small strong $t$-treewidth backdoor exists. 

A natural drawback of strong backdoors, as far as using them for solving \textsc{SAT} is concerned,
is that
they require \emph{every} partial assignment of $X$ to yield a bounded-treewidth instance. 
In practice, to obtain a satisfying assignment, we merely need a set $X$ of $k$ variables such that there is \emph{at least one} assignment $\alpha:X\rightarrow \{0,1\}$ under which the given formula $\varphi$ is satisfiable and of treewidth at most~$t$, provided that the unrestricted $\varphi$ is satisfiable. 
Such a set $X$ is called a \emph{weak $t$-treewidth backdoor}. 
Intuitively, while strong backdoors act as witnesses for tractability, weak backdoors can be seen as direct witnesses for satisfiability.

Unfortunately,
in contrast to strong backdoors,
 finding even a weak backdoor to treewidth 1
 for an arbitrary formula $\varphi$ is $\W[2]$-hard with respect to $k$~\cite{DBLP:conf/icalp/GaspersS12}.
A natural response to this hardness is to restrict the structure of $\varphi$; in fact, 
subsequent research has established that finding weak $t$-treewidth backdoors becomes tractable when the input formulas are {\em structurally sparse}. In particular, 
Fomin, Lokshtanov, Misra, Ramanujan, and Saurabh~\cite{DBLP:conf/soda/FominLMRS15} gave the first fixed-parameter algorithm for solving $d$-\textsc{SAT} via weak $t$-treewidth backdoors by showing that such backdoors can be computed in linear time for every fixed choice of $k$, $t$, and $d$. 
More recently, Lokshtanov, Panolan, and Ramanujan~\cite{DBLP:conf/icalp/LokshtanovP022} expanded this frontier via an $f(k,t)\cdot |\varphi|^{\mathcal{O}(d)}$ 
algorithm for computing weak $t$-treewidth backdoors on any $K_{d,d}$-free \textsc{SAT} instance $\varphi$. This highly general result applies
to all previously studied sparse graph classes, including graphs of bounded degree, bounded-genus graphs, and nowhere-dense graph classes.

Given the above results, a natural follow-up question is whether the tractability of \textsc{SAT} via weak $t$-treewidth backdoors can be pushed beyond sparsity into the regime of ``well-behaved'' dense formulas. 
In our work, we answer this question positively by giving a linear-time algorithm for solving \textsc{SAT} via weak $t$-treewidth backdoors on instances that are {\em signed $d$-semi-ladder-free}.
In the next paragraph, we define and motivate the study of these graph classes.

\subparagraph{Beyond Sparsity.}
Our search for tractable structures in the dense regime is inspired by the line of work that lifts first-order (FO) model checking from nowhere dense classes to their dense counterparts (see~\cite{DBLP:journals/corr/abs-2501-04166} for a survey).
A recurring combinatorial notion in this line of work is the {\em ladder}.
A {\em $d$-ladder}, also known as a half-graph, is a graph with vertices $a_1, \ldots, a_d$ on one side and vertices $b_1, \ldots, b_d$ on the other, and there is an edge $\{a_i,b_j\}$ if and only if $i \leq j$.
We refer to $d$ as the height of the ladder.
A graph class is {\em edge-stable} if it avoids some $d$-ladder as a semi-induced subgraph, and much of the recent progress on FO model checking beyond sparsity concerns dense classes whose structure is governed by the absence of ladders in logically defined auxiliary graphs (monadic stability)~\cite{DBLP:conf/focs/DreierEMMPT24}.\footnote{We note that our problem cannot be expressed as an FO model checking problem.}
Excluding ladders in the input graph itself, however, is a purely combinatorial condition that is considerably more general than such logically defined conditions, and it is natural to ask whether it already suffices for the detection of weak $t$-treewidth backdoors.
It does not: we show in Theorem~\ref{thm:lowerbounds-ladderfree} that the problem is \W[1]-hard on $d$-ladder-free, that is, edge-stable, instances.

{\em Semi-ladders} arise as a natural relaxation of the notion of ladders. 
While a ladder prescribes both the forced edges and the forced non-edges between the two ordered sides of a bipartition, a semi-ladder keeps only the forced edges and additionally requires horizontal non-edges;
that is, 
we have the edge $\{a_i,b_j\}$ whenever $i <j$, and
the non-edges $\{\{a_1,b_1\}, \ldots, \{a_d,b_d\}\}$. See Figure~\ref{fig:ladders} for a visualization of both notions.
Thus, when excluding a semi-ladder rather than an induced ladder (half-graph), the ``lower triangle'' is left unconstrained. Notice that since every ladder contains a semi-ladder (up to $\pm 1$ in the height) but the converse does not hold, the class of $d$-ladder-free graphs is strictly more general than the class of $d$-semi-ladder-free graphs, making the latter a promising candidate for delineating the boundary of efficient weak backdoor detection.

The relevance of semi-ladders can also be understood through the {\em Helly property}---see for example~\cite{DBLP:conf/stacs/FabianskiPST19}.  
In a bipartite incidence graph $G=(L,R)$, 
where vertices of $L$ represent possible choices and vertices of~$R$ represent constraints, 
a Helly-type property says that whenever a set of constraints cannot be covered by a single choice from $L$, this failure is already witnessed by a bounded-size subset of constraints. Equivalently, the family of candidate sets $N(b)\cap L$, $b \in R$, has bounded-size certificates for empty intersection. 
The obstruction to such a strong Helly property is precisely a large co-matching~\cite{DBLP:conf/stacs/FabianskiPST19}. 
Thus, since bounded semi-ladder index is equivalent, up to Ramsey bounds, to simultaneously bounding the ladder height and co-matching index~\cite{DBLP:conf/stacs/FabianskiPST19}, 
semi-ladder-freeness can be viewed as the conjunction of stability (excluding order-like half-graphs) and a Helly-type finite-obstruction principle (excluding large co-matchings).
This makes it especially natural for domination, hitting, and \textsc{SAT}-backdoor problems, where the algorithmic task is to cover constraints by a small number of selected vertices or literals. Moreover, fixed-radius neighborhood relations in nowhere-dense graph classes have bounded semi-ladder index, showing that semi-ladder-freeness captures a structural phenomenon already present in sparsity theory while extending it to dense graph classes such as map graphs~\cite{DBLP:conf/stacs/FabianskiPST19} and existential-positive first-order interpretations in nowhere-dense graph classes~\cite{MahlS26}.

Unfortunately, it turns out that simply excluding $d$-semi-ladders in the instance still yields \W[2]-hardness for the problem of detecting weak $t$-treewidth backdoor sets, as we will show in Theorem~\ref{thm:lowerbounds-unsigned}.
The reason here, though, is subtle.
The hardness occurs because the incidence graph only captures the complexity of a ``monotone'' \textsc{SAT} instance. That is, the edges of the incidence graph do not account for the polarity in which literals occur in clauses.

One way to circumvent this issue, which is
in line with past research on algorithms for dense \textsc{SAT}~\cite{FischerMR08,SlivovskyS13,GanianHO13,PaulusmaSS16,WallonM20}, is to consider the {\em signed incidence graph} of the input formula. Here, every edge between a vertex that corresponds to a clause and a vertex that corresponds to a variable carries a positive or negative sign depending on how this variable appears in the clause (cf.\ Figure~\ref{fig:signed_unsigned}).
Then, the notion of a signed $d$-semi-ladder-free instance means that a $d$-semi-ladder is not present in the incidence graph restricted to positive edges and in the incidence graph restricted to negative edges (cf.\ Figure~\ref{fig:ladders}).
It is not difficult to show that this is, in fact, equivalent (up to a factor of $2$ and an additive constant) to excluding the $d$-semi-ladder in the edge-unlabeled graph, where each variable is represented by two vertices: one for its positive and the other for its negative occurrences.
This notion of signed
$d$-semi-ladder-freeness turns out to be exactly the right condition for making the detection of weak $t$-treewidth backdoors tractable. Proving this constitutes the main technical contribution of our paper (cf.\ Theorem~\ref{thm:main}).

On a different note,
one can also observe that the (unsigned) incidence graph of a \textsc{SAT} instance 
is precisely the incidence graph of the set system defined via the occurrence of variables in clauses.
This viewpoint makes it natural to turn to structural parameters that have been widely studied for set systems and try to exploit them for solving \textsc{SAT}.
One such parameter that has been employed in the study of set systems is the \emph{VC-dimension}~\cite{vC71,DBLP:journals/jct/Sauer72,Shelah72}; it is well known (and also fairly easy to observe) that bounding the height of the largest (semi-)ladder in the incidence graph also bounds the system's VC-dimension. 
In this sense, semi-ladders are a natural next step from both the graph-theoretic and set-theoretic viewpoints on \textsc{SAT}.

\subparagraph{Our Contribution.} 
As elaborated earlier, in an attempt to relax structural input restrictions
 from nowhere-dense graph classes to edge-stable graph classes or
 to bounded VC-dimension set systems, 
    we identify {\em signed $d$-semi-ladder-freeness} as the correct dense structural restriction that delineates the boundary of tractability.
We establish the following theorem.

\begin{restatable}{theorem}{findingBackdoor}
\label{thm:main}
For all non-negative integers $d,k$, and $t$, there is a linear-time algorithm that takes as input a signed $d$-semi-ladder-free \textsc{SAT} instance $\varphi$ and either:
\begin{itemize}
\item outputs a satisfying assignment for $\varphi$ along with a weak $t$-treewidth backdoor of size at most~$k$, or
\item detects that $\varphi$ contains no weak $t$-treewidth backdoor of size at most $k$.
\end{itemize}

\end{restatable}
As highlighted in the lower part of Figure~\ref{fig:classes} and discussed in the next paragraph, Theorem~\ref{thm:main} generalizes the linear-time fixed-parameter algorithm for \textsc{SAT} via weak $t$-treewidth backdoors on $d$-\textsc{SAT}~\cite{DBLP:conf/soda/FominLMRS15} to a more general class of instances. It also improves upon the more recent fixed-parameter algorithm via weak $t$-treewidth backdoors on $K_{d,d}$-free \textsc{SAT}, both by applying to an even broader class of instances and by avoiding the superlinear factor of $n^{\mathcal{O}(d)}$ in the running time. 

We remark that a weak $0$-treewidth backdoor corresponds to a partial assignment that satisfies all clauses. As a special case, we thus obtain the following corollary.
\begin{cor}\label{cor:zero-tw-backdoor}
    For all non-negative integers $d$ and $k$, there is a linear-time algorithm which takes as input a signed $d$-semi-ladder-free \textsc{SAT} instance $\varphi$ and either:
    \begin{itemize}
    \item outputs a partial assignment of at most $k$ variables satisfying all clauses in $\varphi$, or
    \item reports that no such partial assignment exists.
    \end{itemize}
\end{cor}

On a related note, the tractability established by Corollary~\ref{cor:zero-tw-backdoor} contrasts with a problem known as {\sc Minimum Weight SAT}. 
In {\sc Minimum Weight SAT}, given a \textsc{SAT} formula and a positive integer $k$, the goal is to find a satisfying assignment that sets at most $k$ variables to true, or correctly determine that no such assignment exists. Interestingly, one can show that this problem is \W[1]-hard even on $K_{2,2}$-free instances, and hence when excluding (signed) semi-ladders of constant height (see Lemma~\ref{lem:min-wt-sat} in the Appendix). 
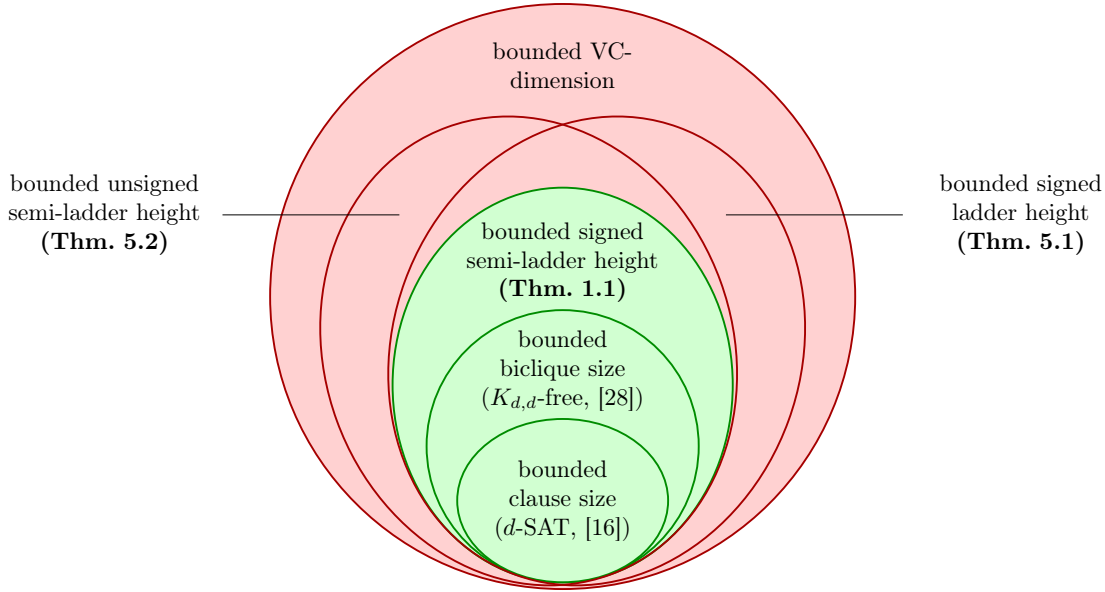
\begin{figure}[!htb]
    \centering
    \scalebox{0.9}{\begin{tikzpicture}[
    every node/.style={align=center, font=\small},
    tractfill/.style={fill=green!18, draw=green!55!black, thick},
    hardline/.style={draw=red!65!black, thick},
    leader/.style={->, thin, gray!55!black, shorten >=2pt, shorten <=2pt},
  ]
  \fill[red!18] (0,4.2) ellipse [x radius=4.3cm, y radius=4.3cm];
  % --- green nest (slim), drawn over the red ---
  \draw[tractfill] (0,2.9) ellipse [x radius=2.5cm, y radius=2.9cm]; % c3
  \draw[tractfill] (0,2.0) ellipse [x radius=2.0cm, y radius=2.0cm]; % c2
  \draw[tractfill] (0,1.2) ellipse [x radius=1.55cm, y radius=1.2cm]; % c1
  % --- red outlines on top ---
  \draw[hardline, rotate around={ 20:(-0.5,3.4)}] (-0.5,3.4) ellipse [x radius=3cm, y radius=3.5cm]; % E
  \draw[hardline, rotate around={-20:(0.5,3.4)}] (0.5,3.4) ellipse [x radius=3cm, y radius=3.5cm]; % S
  \draw[hardline] (0,4.2) ellipse [x radius=4.3cm, y radius=4.3cm];
  % --- inner labels ---
  \node[text width=2.8cm] at (0,1.2)
    {bounded \\ clause size\\(\mbox{$d$-SAT}, \cite{DBLP:conf/soda/FominLMRS15})};
  \node[text width=3.2cm] at (0,3.1)
    {bounded \\ biclique size\\($K_{d,d}$-free, \cite{DBLP:conf/icalp/LokshtanovP022})};
  \node[text width=3.8cm] at (0,4.7)
    {{bounded signed semi-ladder height}\\\textbf{(Thm.~\ref{thm:main})}};
  \node[text width=3.2cm, anchor=east] (lblE) at (-5,5.4)
    {{bounded unsigned semi-ladder height}\\\textbf{(Thm.~\ref{thm:lowerbounds-unsigned})}};
  \draw (lblE.east) -- (-2.4,5.4);
  % --- S label: outside right, leader into the right lobe ---
  \node[text width=3.2cm, anchor=west] (lblS) at (5,5.4)
    {{bounded signed ladder height}\\\textbf{(Thm.~\ref{thm:lowerbounds-ladderfree})}};
  \draw (lblS.west) -- (2.4,5.4);
  % --- top label ---
  \node[text width=3.2cm] at (0, 7.6) {bounded VC-\\dimension};
\end{tikzpicture}}
      \caption{The complexity of finding weak $t$-treewidth backdoors in increasingly more general formulas. Classes that are fixed-parameter tractable with respect to the size of the backdoor are highlighted in green, intractable classes in red. Most of the depicted inclusions are well-known. Additionally, Lemma~\ref{lem:unsigned_signed_correspondence} proves that every formula with bounded signed semi-ladder height also has bounded unsigned semi-ladder height.}
    \label{fig:classes}
\end{figure}

\begin{figure}[!htb]
    \centering
   \scalebox{0.8}{ \begin{tikzpicture}[
    every node/.style={circle,draw,fill=white,inner sep=2pt},
    yscale=1.1, xscale=1.6
  ]
  %Semi Ladder
  % vertices on the left (a) and right (b)
  \foreach \i in {1,2,3,4}{
    \node (a\i) at (0,5-\i) {};
    \node (b\i) at (2,5-\i) {};
  }
  % labels
  \foreach \i in {1,2,3,4}{
    \node[draw=none,fill=none,left=4pt of a\i] {$a_{\i}$};
    \node[draw=none,fill=none,right=4pt of b\i] {$b_{\i}$};
  }
  % dashed horizontal guide lines a_i -- b_i
  \foreach \i in {1,2,3,4}{
    \draw[dashed,gray] (a\i) -- (b\i);
  }
  % solid edges
  \draw (a1) -- (b2);
  \draw (a1) -- (b3);
  \draw (a1) -- (b4);
  \draw (a2) -- (b3);
  \draw (a2) -- (b4);
  \draw (a3) -- (b4);

    %Ladder
% vertices on the left (a) and right (b)
  \foreach \i in {1,2,3,4}{
    \node (a\i) at (0+4,5-\i) {};
    \node (b\i) at (2+4,5-\i) {};
  }
  % labels
  \foreach \i in {1,2,3,4}{
    \node[draw=none,fill=none,left=4pt of a\i] {$a_{\i}$};
    \node[draw=none,fill=none,right=4pt of b\i] {$b_{\i}$};
  }

    \draw[dashed,gray] (a2) -- (b1);
    \draw[dashed,gray] (a3) -- (b1);
    \draw[dashed,gray] (a3) -- (b2);
    \draw[dashed,gray] (a4) -- (b1);
    \draw[dashed,gray] (a4) -- (b2);
    \draw[dashed,gray] (a4) -- (b3);
  
  % solid edges
  \draw (a1) -- (b1);
  \draw (a1) -- (b2);
  \draw (a1) -- (b3);
  \draw (a1) -- (b4);
  \draw (a2) -- (b2);
  \draw (a2) -- (b3);
  \draw (a2) -- (b4);
  \draw (a3) -- (b3);
  \draw (a3) -- (b4);
  \draw (a4) -- (b4);
\end{tikzpicture}}
      \caption{A $4$-semi-ladder (left) and a $4$-ladder (right). Here, solid lines indicate edges and dashed lines represent non-edges. Vertices that are not connected by a solid or dashed line may be either adjacent or non-adjacent.
      }
      \label{fig:ladders}
      \vspace{0.5cm}
      \scalebox{0.8}{\begin{tikzpicture}[
    every node/.style={circle,draw,fill=white,inner sep=2pt},
    yscale=1.1, xscale=1.6
  ]
  %Semi Ladder
  % vertices on the left (a) and right (b)
  \foreach \i in {1,2,3,4,5,6}{
    \node (a\i) at (\i,0) {};
    \node (b\i) at (\i,2) {};
  }
  % labels
  \foreach \i in {1,2,3,4,5,6}{
    \node[draw=none,fill=none,below=4pt of a\i] {$a_{\i}$};
    \node[draw=none,fill=none,above=4pt of b\i] {$b_{\i}$};
  }
  % dashed horizontal guide lines a_i -- b_i
  \foreach \i in {1,2,3,4,5,6}{
    \draw[dashed,gray] (a\i) -- (b\i);
  }
  % solid edges
  \foreach \i in {2,3}{
    \draw (a1) -- (b\i);
  }
   \draw (a2) -- (b3);

  \foreach \i in {5,6}{
    \draw (a4) -- (b\i);
  }
\draw (a5) -- (b6);

 \foreach \i in {4,5,6}{
    \draw[line width=1.2pt, draw=red] (a1) -- (b\i);
  }
 \foreach \i in {4,5,6}{
    \draw[line width=1.2pt, draw=red] (a2) -- (b\i);
  }
\foreach \i in {4,5,6}{
    \draw[line width=1.2pt, draw=red] (a3) -- (b\i);
}

\end{tikzpicture}}
      \caption{A $6$-semi-ladder contains a $K_{3,3}$ as a subgraph.}
    \label{fig:ladders_and_kdd}
    \vspace{0.5cm}
    \begin{tikzpicture}[
    every node/.style={font=\small},
    var/.style={                 % variable / literal vertex
        draw,
        circle,
        minimum size=0.6cm,
        thick,
        fill=violet!15,
        inner sep=0.5pt
    },
    cls/.style={                 % clause vertex
        draw,
        rectangle,
        minimum size=0.6cm,
        thick,
        fill=teal!18,
        inner sep=2pt
    },
    edg/.style={thick},                      % opposite-literal edge
    posi/.style={thick, green!55!black},              % positive (signed) edge
    neg/.style={thick, red!70!black},                % negative (signed) edge
    poslbl/.style={font=\scriptsize, green!45!black, inner sep=1pt},
    neglbl/.style={font=\scriptsize, red!70!black, inner sep=1pt},
    title/.style={font=\normalsize, align=center, anchor=north}
]

    % ============ panel 1: UNSIGNED ============
    \begin{scope}[shift={(0,0)}]
        \node[var] (x1) at (0.0,0) {$x_1$};
        \node[var] (x2) at (1.4,0) {$x_2$};
        \node[var] (x3) at (2.8,0) {$x_3$};
        \node[cls] (C1) at (0.5,1.8) {$C_1$};
        \node[cls] (C2) at (2.0,1.8) {$C_2$};

        \draw[edg] (C1) -- (x1);
        \draw[edg] (C1) -- (x2);
        \draw[edg] (C2) -- (x1);
        \draw[edg] (C2) -- (x2);
        \draw[edg] (C2) -- (x3);

        \node[title] at (1.4,-0.8) {unsigned: $G_\varphi^{\mathrm{unsign}}$};
    \end{scope}

    % ============ panel 2: SIGNED, edge-labeled ============
    \begin{scope}[shift={(4.5,0)}]
        \node[var] (sx1) at (0.0,0) {$x_1$};
        \node[var] (sx2) at (1.4,0) {$x_2$};
        \node[var] (sx3) at (2.8,0) {$x_3$};
        \node[cls] (sC1) at (0.5,1.8) {$C_1$};
        \node[cls] (sC2) at (2.0,1.8) {$C_2$};

        \draw[posi] (sC1) -- node[poslbl, left]     {$+$} (sx1);
        \draw[posi] (sC1) -- node[poslbl, pos=0.35, left] {$+$} (sx2);
        \draw[neg] (sC2) -- node[neglbl, below=1mm, pos=0.3]  {$-$} (sx1);
        \draw[neg] (sC2) -- node[neglbl, right]    {$-$} (sx2);
        \draw[posi] (sC2) -- node[poslbl, right]    {$+$} (sx3);

        \node[title] at (1.4,-0.8) {signed, edge-labeled: $G_\varphi^{\mathrm{signE}}$};
    \end{scope}

    % ============ panel 3: SIGNED, two literal vertices per variable ============
    \begin{scope}[shift={(9.0,0)}]
        \node[var] (lx1)  at (0.0,0) {$x_1$};
        \node[var] (lnx1) at (0.8,0) {$\overline{x_1}$};
        \node[var] (lx2)  at (1.7,0) {$x_2$};
        \node[var] (lnx2) at (2.5,0) {$\overline{x_2}$};
        \node[var] (lx3)  at (3.4,0) {$x_3$};
        \node[var] (lnx3) at (4.2,0) {$\overline{x_3}$};
        \node[cls] (lC1) at (0.85,1.8) {$C_1$};
        \node[cls] (lC2) at (2.17,1.8) {$C_2$};

        % opposite-literal edges (perfect matching)
        \draw[edg] (lx1) -- (lnx1);
        \draw[edg] (lx2) -- (lnx2);
        \draw[edg] (lx3) -- (lnx3);

        % incidence edges (polarity encoded by which literal vertex)
        \draw[edg] (lC1) -- (lx1);
        \draw[edg] (lC1) -- (lx2);
        \draw[edg] (lC2) -- (lnx1);
        \draw[edg] (lC2) -- (lnx2);
        \draw[edg] (lC2) -- (lx3);

        \node[title] at (2.05,-0.8) {signed, literal vertices: $G_\varphi^{\mathrm{signV}}$};
    \end{scope}
\end{tikzpicture}
    \caption{The distinct incidence graphs of the formula $\varphi=(x_1\vee x_2)\,\wedge\,(\overline{x_1}\vee\overline{x_2}\vee x_3)$. }
    \label{fig:signed_unsigned}
\end{figure}

Crucially, we complement Theorem~\ref{thm:main} by showing that signed semi-ladder-freeness marks the tractability barrier of finding weak $t$-treewidth backdoors---see the top part of Figure~\ref{fig:classes}. More precisely, we first show in Theorem~\ref{thm:lowerbounds-unsigned} that weak $t$-treewidth backdoor detection is \W[2]-hard if we merely exclude $d$-semi-ladders in the incidence graph---i.e., without taking the edge signs into account. Second, in Theorem~\ref{thm:lowerbounds-ladderfree} we prove that weak $t$-treewidth backdoor detection is \W[1]-hard on \textsc{SAT} instances that are signed ladder-free (as opposed to semi-ladder-free). In fact, both lower bounds hold even if $d$ is a small constant and apply to every fixed choice of $t\geq 0$.

\subparagraph{Related Work on Backdoors.}
The study of backdoor sets has developed into an active line of research for understanding hidden structures in satisfiability problems that can be exploited to obtain tractable algorithms.
The general backdoor paradigm was introduced by Williams, Gomes, and Selman~\cite{WilliamsGS03}.
Building on this idea of identifying a small set of variables whose assignment reduces the instance to a tractable class, the literature has considered backdoors to both syntactic classes (such as Horn- and $2$-\textsc{SAT}~\cite{DBLP:conf/sat/NishimuraRS04}, affine~\cite{DBLP:journals/jcss/GaspersMOSZ17}, nested~\cite{DBLP:conf/sat/GaspersS12}, and $q$-Horn formulas~\cite{DBLP:journals/algorithmica/GaspersORSS16}) and to structural classes (such as acyclic formulas~\cite{DBLP:conf/icalp/GaspersS12} and more generally formulas of bounded incidence treewidth~\cite{GaspersS13}). 
Surveys, including a very recent one, by Gaspers and Szeider~\cite{DBLP:journals/csr/GaspersS26, DBLP:conf/birthday/GaspersS12} emphasize that the parameterized complexity of detecting such backdoors has become a broad research program, including variants for strong and weak backdoors, structural backdoors, and extensions beyond classical \textsc{SAT}. 

This activity is ongoing: recent work on quantified Boolean formulas (\textsc{QBF}s) studied backdoors to Horn, 2-\textsc{SAT}, and affine \textsc{QBF}s, showing that the \textsc{QBF} setting is substantially more delicate than \textsc{SAT}---\textsc{QBF} remains \PSPACE-hard even for formulas with constant-size standard backdoors---while introducing enhanced backdoors that recover \textsc{FPT} algorithms under additional structural control such as bounded quantifier depth and bounded incidence treewidth~\cite{DBLP:journals/corr/abs-2604-15927}.

The problem of detecting weak $t$-treewidth backdoor sets for \textsc{SAT}, the one considered in this work, 
 is shown not to admit a polynomial kernel parameterized by $k$, even for $d$-\textsc{SAT} when $d \geq 4$ and $t$ is fixed, under reasonable complexity assumptions~\cite{DBLP:conf/soda/FominLMRS15}.

\subparagraph{Related Work on $d$-semi-ladder-free Graphs.}
Fabianski et al.~\cite{DBLP:conf/stacs/FabianskiPST19}
introduced semi-ladder-freeness in the context of their
progressive exploration framework, obtaining FPT
algorithms for domination problems on several graph classes.
Subsequently, Guillemot~\cite{DBLP:journals/toct/Guillemot25}
established FPT and kernelization
results for \textsc{Set Cover} on $d$-semi-ladder-free hypergraphs. 
A certain part of the algorithm that underlies our main result is inspired by the latter work; we indicate the respective part of the algorithm below.

\section{Technical Overview}\label{sec:tech-overview}

We now present the key ideas and challenges towards proving Theorem~\ref{thm:main}.

\subparagraph{The General Procedure.}
The algorithm that Theorem~\ref{thm:main} postulates is a branching algorithm that orchestrates three major components in an iterative fashion.
We call these components \texttt{reducer}, \texttt{extractor}, and \texttt{progressor}.
 The pipeline that comprises all three components is illustrated in Figure~\ref{fig:one_iteration}.
In broad strokes, the \texttt{extractor}
samples a set $V'$ of variables and clauses 
of bounded size
such that every weak backdoor~$X$, if there is any, intersects $V'$ or satisfies a clause in $V'$ via a partial assignment witnessing its weak backdoor property.
The \texttt{extractor} may not be able to achieve the above on general instances, which is why the instances are pre-processed by the \texttt{reducer} component, which ensures that an instance exhibits a certain structure that guarantees the success of the \texttt{extractor}.
Finally, the \texttt{progressor} takes as input a variable or a clause that is guaranteed to be in the weak backdoor or be satisfied by it, respectively, and eliminates it from the instance, thus ensuring that some progress towards finding a weak backdoor is made. 

To motivate the choice of these modules and their interplay, we first sketch the general procedure for finding weak $t$-treewidth backdoors on instances of $d$-\textsc{SAT} and $K_{d,d}$-free \textsc{SAT} as proposed in~\cite{FominLMRS15_uib} and~\cite{DBLP:conf/icalp/LokshtanovP022}, respectively. 
These works operate on the \emph{unsigned} incidence graph of the respective \textsc{SAT} instances. 

Thus, let $G$ be the unsigned incidence graph of a \textsc{SAT} formula $\varphi$ and let~$X$ be a hypothetical weak $t$-treewidth backdoor of $\varphi$ of size at most $k$. Then, in particular, the treewidth of the residual graph~$G-N_G[X]$ is at most $t$, where $N_G[X]$ denotes the closed neighborhood of~$X$ in~$G$.\footnote{Equivalently, one can say that $N_G[X]$ is a modulator to treewidth $t$ in $G$.} Observe that if the size of $N_G[X]$ were bounded by some function of $k$, then the original instance would have a treewidth of size $f(k,t)$ (where $f$ is some computable function) and we could immediately apply Courcelle's Theorem~\cite{Cou90}. 

Unfortunately, the size of $N_G[X]$ may be large; in fact, this may be the case even if $\varphi$ is a formula of bounded clause size, since the variables in $X$ may appear in many clauses. 
Hence, instead of relying on the size of $N_G[X]$ to be bounded, we further observe that if a ``good'' fraction of the edges of $G$ is not incident on $N_G[X]$, then there are large subgraphs in $G-N_G[X]$, which have treewidth at most~$t$ and a ``boundary'' of size $\mathcal{O}(t)$ to the rest of the graph.
Such subgraphs are called $\mathcal{O}(t)$-protrusions in the literature.
Essentially, the task of our \texttt{reducer} is to compress large $\mathcal{O}(t)$-protrusions in the graph $G$, while preserving the satisfiability status of the instance.

After the instance has been processed by the \texttt{reducer}, a random uniform sampling of an edge of the graph and then a random selection of an end-point of this edge gives a vertex in $N_G[X]$ with high probability. This random sampling---in fact, a deterministic version of it---is the task of the \texttt{extractor}.

Finally, the \texttt{progressor} gets a vertex that is guaranteed to be in $N_G[X]$ with ``high probability'' (this is simulated by branching in the case where the \texttt{extractor} is a deterministic algorithm).
If this vertex corresponds to a variable then the \texttt{progressor} chooses to pick it in the hypothetical weak backdoor~$X$,
guesses an assignment to it, which is compatible with a satisfying assignment of the formula, and then removes all the clauses that are satisfied by this assignment of the chosen vertex. In this case, the \texttt{progressor} made progress by dropping the measure $k$ by $1$.

The interesting case is when the \texttt{extractor} produces a vertex which happens to be a clause vertex in $N_G[X]$. 
In this case, we know that the hypothetical weak backdoor contains a literal that is one of the literals of this clause.
If the input formula was a $d$-\textsc{SAT} formula, then the length of each clause is bounded by $d$, and therefore one can branch on which literal in this clause is in $X$ and then proceed to delete it, delete all clauses it satisfies and decrease $k$ by $1$. 
This is roughly how the algorithm of Fomin et al.~\cite{FominLMRS15_uib} for $d$-\textsc{SAT} operates.
The non-trivial scenario arises when the clauses have unbounded length, which prevents this final branching of guessing which literal of the clause satisfies it in the weak backdoor $X$.

The key obstacle to make this procedure work on instances beyond $d$-\textsc{SAT}, 
is having a notion of a bounded measure that can still drop while eliminating a sampled clause from $N_G[X]$. We elaborate on this main obstacle below while highlighting the impact of this in the interplay between our three components \texttt{reducer}, \texttt{extractor}, and \texttt{progressor}.
Also note that the \texttt{reducer}
and the \texttt{extractor} function irrespectively of the structure of the input formula.
It is the \texttt{progressor} that needs additional structure on the input, in order to make ``progress''.

The earlier work of Lokshtanov et al.~\cite{DBLP:conf/icalp/LokshtanovP022} for $K_{d,d}$-free \textsc{SAT} also encounters this challenge and manages to resolve it, in a way, by waiting for several rounds of the ``\texttt{extractor}-\texttt{progressor} communication'' (still bounded by a function of $k,t,d$) before really dropping the measure of the branching algorithm. 
This is not quite the formal treatment; instead the authors make use of the notion of unbreakability of instances to overcome the challenge.
The implementation of this, in the work of Lokshtanov et al.\ costs them super-linear time---in fact the time taken is {\XP} w.r.t.\ $d$ (that is, $|\varphi|^{\mathcal{O}(d)}$). 
Our work, on the other hand, overcomes this issue in {\em linear time} on a even more general class of instances.

\subparagraph{The \texttt{progressor}.}
Our \texttt{progressor}, is inspired by the 
FPT algorithm for {\sc Set Cover} on 
on $d$-semi-ladder-free hypergraphs by Guillemot~\cite{DBLP:journals/toct/Guillemot25}.
On receiving a clause $C$
that is guaranteed to be satisfied by the assignment $\alpha$ of the weak backdoor $X$, 
records that some literal of this clause must be selected into the weak backdoor. 
This information is kept by imagining $X$ as an ordered set of $k$ literals and maintaining $k$ (not necessarily disjoint) sets of literals, called {\em candidate sets} $P_1, \ldots, P_k$, that keep the choices for the first, second, ..., $k$-th literal of the backdoor, respectively.
In other words, in this annotated problem, the $i$-th backdoor variable is forced to come from the set $P_i$. Initially, each $P_i$ contains all literals.
When the \texttt{progressor} receives a clause $C$ that is guaranteed to be satisfied by the weak backdoor, the \texttt{progressor} guesses the least index $i$ such that $C$ will be satisfied by the $i$-th literal of the weak backdoor.
Under the guess, it deletes the $i$-th candidate set $P_i$ to only contain the literals of $C$ that were already in $P_i$, and delete everything else.
This way, the new $P_i$ ensures that as long as, for each $i$, some literal of $P_i$ is selected into the final weak backdoor, the clause $C$ will always be satisfied.

To guarantee termination of the algorithm, we will show that some measure decreases every time the \texttt{progressor} is called. To this end, the algorithm cleans up the instance after every branching by deleting all clauses whose neighborhood contains some $P_i$. 
This deletion is correct since,
as long as any literal from each set $P_i$ is chosen to be satisfied by the final weak backdoor,
the deleted clauses will all be satisfied.
The clean-up ensures that after each branching step followed by a clean-up,
the size of a maximum semi-ladder that can be constructed using the literal vertices of $P_i$ (and arbitrary clauses) in the signed incidence graph,  
strictly decreases. Thus, the total measure of all branching steps of the \texttt{progressor} is bounded by $s:= (d-1)k$. 
This bounds the number of calls to the \texttt{progressor} in the overall algorithm. See Figure~\ref{fig:fig_highlevel}.

\subparagraph{The \texttt{reducer}.}
Finally, since the \texttt{progressor} makes progress by reducing a bounded measure at the cost of adding many annotations to the instance, it correspondingly increases the burden on the \texttt{reducer} routine.

Recall that the \texttt{reducer} is a pre-processing routine that outputs an equivalent instance that is free from large $\mathcal{O}(t)$-protrusions. Because of the steps outlined above, the \texttt{reducer} now needs to deal with annotated instances. 
The problem, however, is that our annotations, that is, the candidate sets $P_1, \ldots, P_k$ and a measure $s$ associated with the instance, are not encoded in the graph structure.
Consequently, even if a graph contains an $\mathcal{O}(t)$-protrusion---so, in particular, a subgraph with an $\mathcal{O}(t)$-sized boundary---this boundary is not the protrusion's ``real'' one, as information flows between this protrusion and the remaining graph through the annotations, which the graph structure does not capture.
This makes handling protrusions non-trivial.
We circumvent the issue by designing a \texttt{reducer} that efficiently compresses sufficiently large $\mathcal{O}(t)$-protrusions to ones of size $f(k,t,d)$---this differs from the $f(t)$ bound one would obtain when working with the ``genuine'' boundary of $t$-protrusions, but is still sufficient for our purposes.

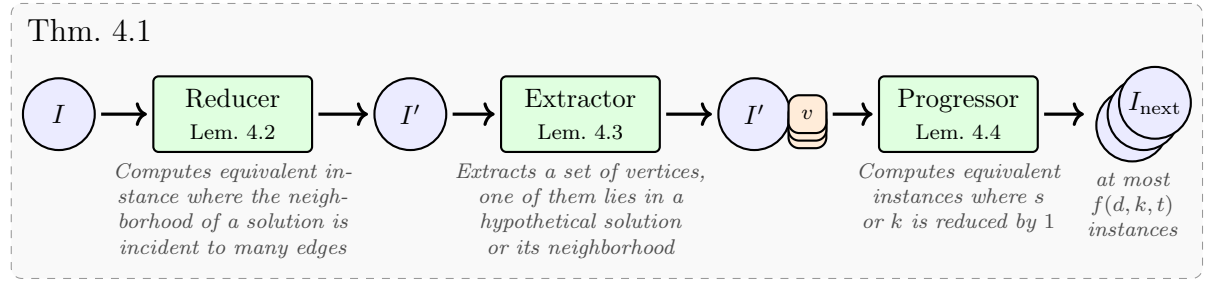
\begin{figure}[!htb]
    \centering
    \begin{tikzpicture}[
    every node/.style={font=\small},
    every text node part/.style={align=center},
    box/.style={
        draw,
        rounded corners=2pt,
        minimum width=2.1cm,
        minimum height=1cm,
        align=center,
        thick,
        fill=green!12,
    },
    groupbox/.style={rounded corners=4pt, draw=gray!80, dashed, inner sep=4pt},
    instance/.style={
        draw,
        circle,
        minimum size=0.95cm,
        thick,
        fill=blue!8,
        inner sep=1pt
    },
    info/.style={
        draw,
        rounded corners=10pt,
        minimum width=1.6cm,
        minimum height=0.8cm,
        align=center,
        thick,
        fill=orange!15,
    },
    smallclause/.style={
        draw,
        rounded corners=4pt,
        minimum size=0.5cm,
        thick,
        fill=orange!15,
        inner sep=1pt,
        font=\scriptsize
    },
    arr/.style={->, thick, shorten >=2pt, shorten <=2pt},
    edgelbl/.style={font=\scriptsize, inner sep=2pt},
    moduledesc/.style={font=\scriptsize\itshape, align=left, text width=3cm, color=black!70}
]
    % Initial instance
    \node[instance] (I0) at (0, 0) {$I$};
    \node[box] (reducer)   at (2.3, 0) {Reducer\\{\scriptsize Lem.~\ref{lem:reducer}}};      
    \node[box] (extractor)  at (6.9, 0)   {Extractor\\{\scriptsize Lem.~\ref{lem:extractor}}};
    \node[box] (progressor) at (11.9, 0)  {Progressor\\{\scriptsize Lem.~\ref{lem:branchingReduction}}};      
    % Placeholder descriptions next to each module
    \node[moduledesc, anchor=north, text width=4cm] (desRed) at ($(reducer.south)+(0,0)$)
        {Computes equivalent instance where the neighborhood of a solution is incident to many edges};
    \node[moduledesc, anchor=north, text width=3.5cm] (desExt) at ($(extractor.south)+(0,0)$)
        {Extracts a set of vertices, one of them lies in a hypothetical solution or its neighborhood};
    \node[moduledesc, anchor=north] (desPro) at ($(progressor.south)+(0,0)$)
        {Computes equivalent instances where $s$ or $k$ is reduced by $1$};

    % Intermediate artifacts on the cycle
    \node[instance] (Iprime) at (4.65, 0)  {$I'$};
    \node[instance] (IprimeC) at (9.2, 0) {$I'$};
    \node[smallclause, anchor=west] (clauseCpp) at ($(IprimeC.east)+(-0.05,-0.2)$) {};
    \node[smallclause, anchor=west] (clauseCp) at ($(IprimeC.east)+(-0.05,-0.1)$) {};
    \node[smallclause, anchor=west] (clauseC) at ($(IprimeC.east)+(-0.05,0)$) {$v$};
    \node[instance] (Inextpp)  at (14.2, -0.15) {};
    \node[instance] (Inextp)  at (14.35, 0) {};
    \node[instance] (Inext)  at (14.5, 0.15) {$I_{\text{next}}$};

    \node[moduledesc, anchor=north] (desPro) at ($(Inextpp.south)$)
        {at most\\$f(d,k,t)$\\instances};
 
    % Connections
    \draw[arr, very thick] (I0) -- (reducer);
    \draw[arr, very thick] (reducer) -- (Iprime);
    \draw[arr, very thick] (Iprime) -- (extractor);
    \draw[arr, very thick] (extractor) -- (IprimeC);
    \draw[arr, very thick] (clauseC) -- (progressor);
    \draw[arr, very thick, shorten >=3mm] (progressor) -- (Inextp);
    % invisible spacer creating headroom for the theorem label above the top row
    \node (headroom) at ($(I0.north)+(0,0.7)$) {};
    \begin{scope}[on background layer]
    \node[groupbox, fill=gray!5, fit=(I0)(desRed)(Inext)(headroom)] (bg) {};
\end{scope}
    \node[anchor=north west, font=\large, inner sep=6pt]
         at (bg.north west) {Thm.~\ref{thm:mainTechnicalTheorem}};
\end{tikzpicture}
    \caption{A schematic depiction of the interplay between our algorithmic modules---the \texttt{reducer}, \texttt{extractor}, and \texttt{progressor}---in the proof of Theorem~\ref{thm:mainTechnicalTheorem}.}
    \label{fig:one_iteration}
\end{figure}

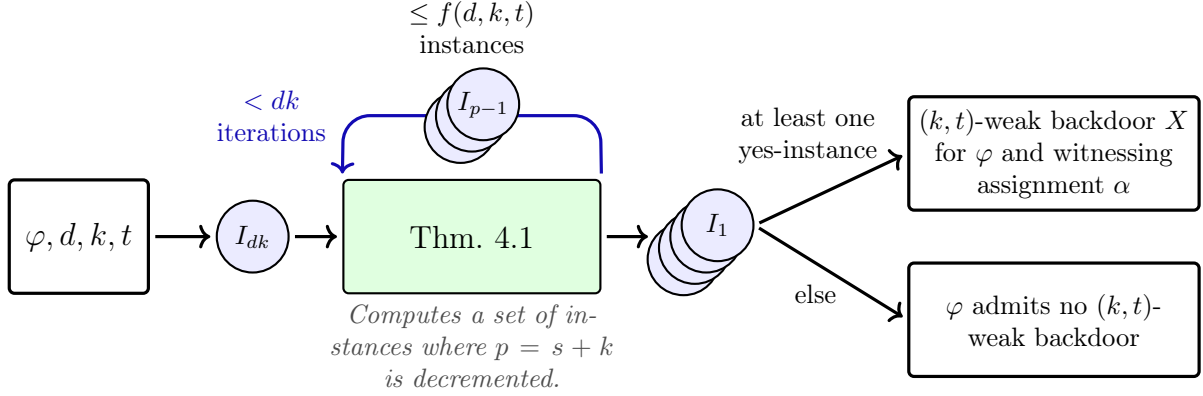
\begin{figure}[!htb]
    \centering
    \begin{tikzpicture}[
    every node/.style={font=\small},
    every text node part/.style={align=center},
    box/.style={
        draw,
        rounded corners=2pt,
        minimum width=2.4cm,
        minimum height=1cm,
        align=center,
        thick,
        fill=green!12,
        font=\large,
    },
    instance/.style={
        draw,
        circle,
        minimum size=0.95cm,
        thick,
        fill=blue!8,
        inner sep=1pt
    },
    output/.style={
        draw,
        rounded corners=2pt,
        align=center,
        very thick,
        fill=white,
        font=\small,
        inner sep=2pt,
        text width=3.7cm,
        minimum height=1.5cm
    },
    arr/.style={->, thick, shorten >=2pt, shorten <=2pt},
    edgelbl/.style={font=\small, inner sep=2pt},
    moduledesc/.style={font=\small\itshape, align=left, text width=4.5cm, color=black!70}
]
    % Input formula and parameters
    \node[output,font=\large,text width=1.7cm] (phi) at (0, 0) {$\varphi, d,k,t$};
    % Initial instance of the extension problem
    \node[instance] (Iinit) at (2.3, 0) {$I_{dk}$};
    % The whole reducer/extractor/progressor pipeline, abstracted as one
    % application of the main technical theorem.
    \node[box, minimum width=3.4cm, minimum height=1.5cm] (mainthm) at (5.2, 0)
        {Thm.~\ref{thm:mainTechnicalTheorem}};
    % The two possible outputs (the solver is folded into the box above)
    \node[output] (outSAT)     at (12.9, 1.1)  {($k,t$)-weak backdoor $X$\\for $\varphi$ and witnessing\\ assignment $\alpha$};
    \node[output] (outNoWBphi) at (12.9, -1.1) {$\varphi$ admits no ($k,t$)-weak backdoor};

    % Straight-through connections
    \draw[arr, very thick] (phi) -- (Iinit);
    \draw[arr, very thick] (Iinit) -- (mainthm);
    % Loop: each application feeds a set of instances back into the next iteration
    \draw[arr, very thick, rounded corners=10pt, color=ultramarine] (mainthm.north east)
        -- (6.9,1.6) -- (3.5,1.6) -- (mainthm.north west);
    % Arrow from the box to the set of output instances
    \draw[arr, very thick] (mainthm.east) -- (7.6,0);

    % Set of instances carried by the loop (front-most labeled I_s)
    \node[instance] (Ilooppp) at (5.05, 1.45) {};
    \node[instance] (Iloopp)  at (5.2, 1.6) {};
    \node[instance] (Iloop)   at (5.35, 1.75) {$I_{p-1}$};
    \node[edgelbl, anchor=south, align=center] at (5.2, 2.35) {$\le f(d,k,t)$ \\ instances};
    \node[edgelbl, anchor=east, align=center] at (3.35, 1.6) {\textcolor{ultramarine}{$< dk$}\\\textcolor{ultramarine}{iterations}};
  %  \node[edgelbl, anchor=south, align=center] at (8.4, 2) {$ < (f(d,k,t))^{dk}$ \\ instances};
    % Set of output instances (front-most labeled I_0)
    \node[instance] (Ioutpp) at (8, -0.3) {};
    \node[instance] (Ioutpp) at (8.15, -0.15) {};
    \node[instance] (Ioutp)  at (8.3, 0) {};
    \node[instance] (Iout)   at (8.45, 0.15) {$I_1$};

    \node[moduledesc, anchor=north] (desThm) at ($(mainthm.south)+(0,0)$)
        {Computes a set of instances where $p=s+k$ is decremented.};

    % Branch to the two outcomes
    \draw[arr, very thick] (Iout.east) -- node[above=4mm, pos=.35, edgelbl, align=center] {at least one\\ yes-instance} (outSAT.west);
    \draw[arr, very thick] (Iout.east) -- node[below=2mm, pos=.4, edgelbl] {else} (outNoWBphi.west);
\end{tikzpicture}
    \caption{High-level overview of the proof of our main result. The input formula $\varphi$ is translated into an annotated instance with a SL-pattern size bound of $s = (d-1)k$ and $p=s+k = dk$. Theorem~\ref{thm:mainTechnicalTheorem} is then applied to this instance to obtain a set of instances in which $k$ or $s$ is reduced by~$1$. The process is iterated on all output instances.
    After at most $dk-1$ iterations we obtain a set of instances with $k=0$ or $s=0$. These can be decided easily, and a solution (or its non-existence) can be lifted back to the original question for the formula~$\varphi$.}
    \label{fig:fig_highlevel}
\end{figure}

\subparagraph{Solving Base Cases}
To conclude the outline of the algorithm, notice that 
once the \texttt{progressor} reduced our ``measure of progress'' sufficiently,
we are guaranteed to remain with a set of ``simple'' instances: for each of them, we have either $k=0$ (asking for the existence of an empty weak backdoor), or $s$, a measure of the ``maximal semi-ladder height of an annotated instance'', is zero. We link the latter to the instance not having any clauses. 
Based on these descriptions, it should be evident that both cases are easy to handle. 
In fact, if $k=0$, it suffices to check whether the instance already has treewidth at most $t$ and, if it does, to determine whether a satisfying assignment exists.
In the second case, that is, if no clauses remain, note that removing any set of variables yields treewidth $0$, so we only have to decide whether such a set exists that respects the annotated information of our instance (in particular, that satisfies the clauses which we removed earlier). We solve this via a simple branching algorithm. 

We have thus presented all ingredients of the algorithm that Theorem~\ref{thm:main} postulates. The overall procedure is depicted in Figure~\ref{fig:fig_highlevel}.
 Below, we highlight several further aspects that make finding weak backdoors in a dense regime more challenging than in the sparse settings studied in~\cite{FominLMRS15_uib} and in~\cite{DBLP:conf/icalp/LokshtanovP022}, respectively.

\subparagraph{The (deterministic) \texttt{extractor}: Structure-oblivious deterministic sampling.}
Earlier works rely on a random sampling argument for the \texttt{extractor} routine. In principle, the randomized argument works for all instances of \textsc{SAT}, as long as they have been pre-processed by the \texttt{reducer}. 
In earlier works, this random step is derandomized by accounting for the sparse structure of the instance. 
In our work, we derandomize without relying on any structure of the \textsc{SAT} formula. In fact, this makes both our \texttt{reducer} and our (deterministic) \texttt{extractor} oblivious to the structure of the input formula.

Below, we briefly sketch the idea of our deterministic \texttt{extractor}.
We first make a fairly simple observation: since a $K_{t+2,t+2}$ is an obstruction to treewidth at most $t$, any $K_{t+2,t+2}$ in the incidence graph $G$ must either be hit by the weak backdoor $X$ or contain a clause vertex that will be satisfied by the witnessing assignment for $X$.
Therefore, if we find a $K_{t+2,t+2}$ in $G$, 
the \texttt{extractor} can output the literals corresponding to all the variables, and the clauses, contained in this $K_{t+2,t+2}$.

The challenge is to find such a $K_{t+2,t+2}$, if it exists, in linear time.
Since finding a biclique is \W[1]-hard parameterized by the desired size~\cite{DBLP:journals/jacm/Lin18}, a linear-time algorithm is out of reach---though finding such a biclique (or reporting that none exist) in $|\varphi|^{\mathcal{O}(t)}$ time is trivial.

We overcome this challenge by exploiting the fact that, due to the pre-processing in the \texttt{reducer} module, any weak backdoor $X$ is such that $N_G(X)$ is a $\rho$-cover of the incidence graph $G$, where $\rho= 1/ f(k,t,d)$ and $f$ is a function that depends only on $k,t,d$. A $\rho$-cover in a graph is a set of vertices that has a $\rho$-fraction of edges incident to it.
Since $N_G(X)$ is a $\rho$-cover, either there is a clause vertex in the neighborhood of $X$ that has at least a $\rho/2dk$-fraction of the edges incident to it, or $X$ contains a literal vertex whose clause neighborhood covers at least a $\rho/k$-fraction of all edges.
In the first case, the number of (clause) vertices which have a $\rho/2dk$-fraction of the edges incident on them individually is bounded by $2dk/\rho$.
In the second case, we observe that, if the number of literal vertices $Y$ whose clause neighborhoods have a 
$\rho/k$-fraction of the edges incident on them is large enough, then there is a $K_{t+2,t+2}$ whose literal endpoints reside in any sufficiently large (but bounded) subset of $Y$. This allows us to find a $K_{t+2,t+2}$ in linear time when $Y$ is large.

\subparagraph{Linear versus almost-linear time in the dense regime.}
Apart from the challenges discussed above, 
yet another issue arises from entering the dense regime:
The \texttt{reducer} is an $ \mathcal{O}(n+m)$-time procedure that reduces the number of vertices in the incidence graph by a fraction of the original number of vertices (if the instance is not already ``ready'' for the \texttt{extractor}). 
Here, $n$ and $m$ are the number of vertices and edges in the incidence graph of the \textsc{SAT} formula.
But one round of this pre-processing may not be enough for the \texttt{extractor} to work, so this step needs to be exhaustively performed until we are in a scenario that the \texttt{extractor} accepts. 
In the sparse regime, this is not a problem if we want to design a linear-time algorithm, since in each time step, the number of vertices reduces by a fraction
and the number of edges is in any case bounded linearly in the number of vertices. 
Therefore, the total running time of this exhaustive application is $\mathcal{O}(n+m) + \mathcal{O}((n+m) /\alpha) + \mathcal{O}((n+m)/\alpha^2) + \ldots$ , for some constant $\alpha > 1$, which is still $\mathcal{O}(n+m)$.

On dense instances, however, a reduction in the number of vertices by a fraction does not guarantee that the number of edges also drops by that proportion. 
Thus, a naive exhaustive application of a reduction rule that eliminates a fraction of the vertices will give an algorithm with running time $\mathcal{O}(|\varphi| \log |\varphi|)$.
Designing a truly linear-time algorithm with these repetitive procedures requires making careful observations about the structure of the problem and the process.

\section{Preliminaries}
\label{sec:prelims}
For any variables $n, \alpha, \beta, \ldots$ we define $\mathcal{O}_{\alpha, \beta, \ldots}(f(n))$ for any function $f$ depending only on $n$ to capture all functions bounded from above by $\mathcal{O}(g(\alpha, \beta, \ldots)\cdot f(n))$ for some computable function $g$. In particular, $\mathcal{O}_{\alpha, \beta, \ldots}(f(n)) = \mathcal{O}(f(n))$ when $\alpha, \beta,\ldots \in \mathcal{O}(1)$.

\paragraph{Graphs.} 
For a graph $G$ we denote its vertex and edge set by $V(G)$ and $E(G)$, respectively. The \emph{size} of a graph is measured by $|G| = |V(G)|+|E(G)|$.
For some $v\in V(G)$, we let $N_G(v) = \{u \mid \{u,v\}\in E(G)\}$ and $N_G[v]= \{v\} \cup N_G(v)$. Moreover, we let $\deg_G(v)= |N_G(v)|$ denote the degree of $v$ in $G$. For a vertex set $S \subseteq V(G)$, we let $N_G(S)=\bigcup_{v\in S} N_G(v) \setminus S$ and $N_G[S]=N_G(S)\cup S$, and we let $G - S$ denote the graph that is obtained from $G$ by deleting all vertices in $S$ together with all edges incident to~$S$.
Throughout the paper, we assume a graph encoding that supports the following operations in $\mathcal{O}(1)$ time: adding/deleting an isolated vertex, adding/deleting a given edge, and testing adjacency between two vertices; moreover, we assume that all edges incident to a vertex can be listed in time linear in the number of such edges.\footnote{Adding and deleting vertices and edges in this way is supported e.g.\ by doubly-linked adjacency lists with cross-pointers between the two occurrences of each edge; note that deleting a non-isolated vertex $v$ then costs $\mathcal{O}(\deg_G(v))$, as one first deletes its incident edges. Edge deletion is possible in $\mathcal{O}(1)$ time given a pointer to one of the two occurrences in the adjacency lists.
To additionally support constant-time adjacency tests, each vertex can maintain a hash set of its neighbors---the complexity of the query, add, and remove operations on this set depends on the machine model, but in general they work in expected $\mathcal{O}(1)$ time (though not necessarily in the worst case).}

The \emph{treewidth} $\tw(G)$ of a graph~$G$ is a well-established measure of how ``tree-like'' it is. For a definition of treewidth and (nice) tree decompositions we refer to the literature~\citep[Subsection 7.2]{CyganFKLMPPS15}.

\begin{definition}[$\rho$-cover]
    Let $G$ be a graph and let $0 < \rho <1$. A set $X \subseteq V(G)$ is a \emph{$\rho$-cover} of $G$ if $\sum_{v \in X} \deg_G(v) \geq \rho \cdot \sum_{v \in V(G)} \deg_G(v)$.
\end{definition}

\paragraph{CNF formulas.}
For a set of Boolean variables $\var$, let $L(\var)=\{x,\bar{x}\mid x\in \var\}$ denote the corresponding set of literals. For simplicity of notation, we define $\bar{\bar{x}}=x$ for all $x\in \var$.
A \emph{clause} over $\var$ is a set of literals from $L(\var)$, and it is  \emph{satisfied} via any of the literals in that set.
A \emph{CNF formula} is a pair $\varphi=(\var,\clause)$, where $\clause$ is a multiset of clauses over $\var$.
Given a CNF formula $\varphi=(\var,\clause)$ and a partial assignment $\alpha\colon X\to \{0,1\}$ for some $X\subseteq \var$, we define the \emph{residual} of $\varphi$ under $\alpha$ to be $\varphi[\alpha]=(\var_\alpha, \clause_\alpha)$ with $\var_\alpha = \var \setminus X$ and $\clause_\alpha = \{C\setminus L(X)\mid C\in \clause \text{ and $C$ is not already satisfied by a literal described by $\alpha$}\}$.
The length of a clause is the number of literals it contains and the \emph{size} of a formula $\varphi$, denoted by $|\varphi|$, is the total length of all clauses in~$\varphi$. 
Below, we define several options for representing  CNF formulas via graphs.

\begin{definition}[Signed Incidence Graphs]
    Let $\varphi=(\var, \clause)$ be a CNF formula. We define the following variants of signed incidence graphs for $\varphi$.
    \begin{itemize}
        \item The \emph{signed vertex incidence graph} $G_\varphi^{\mathrm{signV}}=(\clause \cup L(\var), E)$ contains for each $\ell\in L(\var)$ and $C \in \clause$, the edge $\{\ell, C\}$ if and only if $\ell \in C$, and additionally the edge $\{x,\bar{x}\}$ for each $x\in \var$ that connects two opposing literals.
        \item The \emph{signed edge incidence graph} $G^{\mathrm{signE}}_\varphi=(\clause \cup \var, E)$ contains for each $x\in \var$ and $C \in \clause$, the edge $\{x, C\}$ of sign $+$ if and only if $x \in C$ and of sign $-$ if and only if $\bar x \in C$.
    \end{itemize}
\end{definition}
For $s\in\{+,-\}$, we write $G^{\mathrm{signE},s}_\varphi$ for the subgraph of $G^{\mathrm{signE}}_\varphi$ obtained by keeping all vertices and exactly the edges of sign $s$.

\paragraph{Semi-Ladders.}
A \emph{semi-ladder of height $h$} in a graph $G$ is a sequence of vertices $a_1, \ldots, a_h,$ $b_1, \ldots, b_h$ such that $\{a_i,b_i\}\notin E(G)$ for all $i\in [h]$, and $\{a_i,b_j\}\in E(G)$ for all $i,j\in [h]$ with $i<j$.
Note that the vertices $b_1, \ldots, b_h$ are pairwise distinct (for $i<j$, $b_j$ is a neighbor of $a_i$ while $b_i$ is not) and so are $a_1, \ldots, a_h$. Moreover, restricting a semi-ladder to any subset of its positions, that is, to $a_{i_1}, \ldots, a_{i_m}, b_{i_1}, \ldots, b_{i_m}$ for indices $i_1 < \cdots < i_m$, yields again a semi-ladder.
For $d\in \mathbb{N}$, we say that $G$ is \emph{$SL_d$-free} if it contains no semi-ladder of height at least $d$.

Consider a graph $G$ in which every vertex is marked either as a \emph{literal vertex} or as a \emph{clause vertex}; for instance, in $G_\varphi^{\mathrm{signV}}$ the literal vertices are $L(\var)$ and the clause vertices are $\clause$. We call a semi-ladder $a_1, \ldots, a_h, b_1, \ldots, b_h$ in such a graph \emph{proper} if every $a_i$ is a literal vertex and every~$b_i$ is a clause vertex. Similarly, we call a semi-ladder in $G_\varphi^{\mathrm{signE},s}$ \emph{proper} if every $a_i$ is a variable vertex and every~$b_i$ is a clause vertex.

We have thus defined several notions of semi-ladders in the incidence graphs of a formula; for instance, semi-ladders in $G_\varphi^{\mathrm{signV}}$ may a priori mix literal and clause vertices on both sides, whereas proper ones may not. The following lemma shows that all these notions are functionally equivalent: their maximum heights agree up to a factor of two and an additive constant. Given this correspondence, we restrict our attention to proper semi-ladders in signed vertex incidence graphs in the following and, for simplicity, refer to these graphs as \emph{signed incidence graphs} $G_\varphi$.

\begin{lemma}\label{lem:semi-ladder-notions}
    Let $\varphi$ be a CNF formula. Let $h$ be the maximum height of a semi-ladder in $G_\varphi^{\mathrm{signV}}$, let $h_{V}$ be the maximum height of a proper semi-ladder in $G_\varphi^{\mathrm{signV}}$, and let $h_{E}$ be the maximum height of a proper semi-ladder in $G_\varphi^{\mathrm{signE},+}$ or $G_\varphi^{\mathrm{signE},-}$.
    Then
    $$h_V \leq h \leq h_V + 4 \qquad\text{and}\qquad h_E \leq h_V \leq 2 h_E.$$
\end{lemma}
\begin{proof}
    Since every proper semi-ladder is in particular a semi-ladder, we have $h_V \leq h$.

    We next show $h \leq h_V + 4$. Let $a_1, \ldots, a_h, b_1, \ldots, b_h$ be a semi-ladder in $G_\varphi^{\mathrm{signV}}$ of height $h$. Recall that in $G_\varphi^{\mathrm{signV}}$ no two clause vertices are adjacent, and every literal vertex is adjacent to exactly one other literal vertex, namely the opposing literal. We will moreover use \emph{mirroring}: if $a_1, \ldots, a_h, b_1, \ldots, b_h$ is a semi-ladder, then so is $b_h, \ldots, b_1, a_h, \ldots, a_1$, obtained by exchanging the two sides and reversing their order.
We call a position $i \in [h]$ \emph{proper} if $a_i$ is a literal vertex and $b_i$ is a clause vertex, and \emph{co-proper} if $a_i$ is a clause vertex and $b_i$ is a literal vertex. Let $p$ and $q$ denote the number of proper and co-proper positions, respectively. Restricting the semi-ladder to its proper positions yields a proper semi-ladder of height $p$, and restricting it to its co-proper positions and mirroring yields one of height $q$, so $h_V \geq \max(p,q)$. It therefore suffices to show $\max(p,q) \geq h-4$, which we do by bounding the number of positions of all remaining types.
First, at most one position consists of two clause vertices: given two such positions $i<j$, the clause vertices $a_i$ and $b_j$ would be adjacent.
Second, at most two positions consist of two literal vertices: for two such positions $i<j$, the literal vertices $a_i$ and $b_j$ are adjacent and hence opposing literals, so three such positions $i<j<j'$ would yield $b_j = \overline{a_i} = b_{j'}$, contradicting the distinctness of $b_j$ and $b_{j'}$.
Finally, we claim that $\min(p,q)\leq 1$. Note that every proper position precedes every co-proper position, as otherwise two clause vertices $a_i$ and $b_j$ with $i < j$ would be adjacent. So if $p\geq 2$ and $q\geq 1$, there are proper positions $i<i'$ and a co-proper position $j$ with $i'<j$, and the literal vertex $b_j$ is adjacent to the two distinct literal vertices $a_i$ and $a_{i'}$, a contradiction. Symmetrically, if $p\geq 1$ and $q\geq 2$, there is a proper position $i$ and co-proper positions $j<j'$ with $i<j$, and the literal vertex $a_i$ is adjacent to the two distinct literal vertices $b_j$ and $b_{j'}$, again a contradiction.
In total, $h \leq p + q + 3 \leq \max(p,q) + 4$, as required.

    It remains to show $h_E \leq h_V \leq 2h_E$. Fix a sign $s\in\{+,-\}$ and, for a variable $x\in\var$, let $\ell^s_x = x$ if $s=+$ and $\ell^s_x = \bar{x}$ if $s=-$. By definition of the two graphs, $\{x, C\}$ is an edge of $G_\varphi^{\mathrm{signE},s}$ if and only if $\{\ell^s_x, C\}$ is an edge of $G_\varphi^{\mathrm{signV}}$.

    This correspondence immediately yields $h_V \geq h_E$: if $x_1, \ldots, x_k, C_1, \ldots, C_k$ is a proper semi-ladder in $G_\varphi^{\mathrm{signE},s}$, then $\ell^s_{x_1}, \ldots, \ell^s_{x_k}, C_1, \ldots, C_k$ is a proper semi-ladder in $G_\varphi^{\mathrm{signV}}$.

    Conversely, let $\ell_1, \ldots, \ell_k, C_1, \ldots, C_k$ be a proper semi-ladder in $G_\varphi^{\mathrm{signV}}$ of height $k = h_V$. At least $\lceil k/2 \rceil$ of the literals $\ell_i$ have the same polarity; let $s$ be the corresponding sign and let $i_1 < \cdots < i_m$ with $m \geq \lceil k/2\rceil$ be the corresponding positions. By the same correspondence, the variables of $\ell_{i_1}, \ldots, \ell_{i_m}$ together with $C_{i_1}, \ldots, C_{i_m}$ form a proper semi-ladder in $G_\varphi^{\mathrm{signE},s}$. Hence we have $h_E \geq \lceil h_V/2 \rceil$ and thus $h_V \leq 2 h_E$.
\end{proof}

Another natural graph representation of a CNF formula is the \emph{unsigned incidence graph}, a graph that contains a vertex for each clause and each variable, but does not distinguish between positive and negative occurrences of variables in clauses. Formally, we  define such graphs as follows.
\begin{definition}[Unsigned Incidence Graphs]
    Let $\varphi=(\var, \clause)$ be a CNF formula. The \emph{unsigned incidence graph} $G^{\mathrm{unsign}}_\varphi=(\clause \cup \var, E)$ contains for each $x\in \var$ and $C \in \clause$, the edge $\{x, C\}$ if and only if $x \in C$ or $\bar x \in C$.
\end{definition}

The maximum height of any semi-ladder in the unsigned incidence graph of a formula $\varphi$ gives us insight into the maximum height of any semi-ladder in its signed incidence graph. The relation is not as tight as the correspondence established above---in fact, it is exponential---but still suffices to conclude that every graph whose maximum \emph{signed} semi-ladder height is bounded by a constant also has a constantly bounded \emph{unsigned} semi-ladder height (see Figure~\ref{fig:classes}).

\begin{lemma}\label{lem:unsigned_signed_correspondence}
    Let $\varphi$ be a CNF formula. Let $h_u$ be the maximum height of a semi-ladder in $G_\varphi^{\mathrm{unsign}}$ and $h_s$ be the maximum height of a semi-ladder in $G_\varphi^{\mathrm{signV}}$. Then, $h_u \leq 4^{h_s}+2$.
\end{lemma}
\begin{proof}
    Consider a semi-ladder of maximum height $h_u$ in $G_\varphi^{\mathrm{unsign}}$. If $h_u \leq 3$, then the statement holds trivially since $h_s \geq 0$.
    Moreover, note that in $G_\varphi^{\mathrm{unsign}}$, no literal vertices (\emph{here variable vertices}) are adjacent to each other and neither are clause vertices.
    Therefore, if $h_u\geq 3$, then $G_\varphi^{\mathrm{unsign}}$
    contains a proper semi-ladder $a_1, \dots, a_{h_u^p}, b_1, \dots, b_{h_u^p}$ of height $h_{u}^p$ where $h_u-2\leq h_u^p \leq h_u$.
    Now, for each pair $i<j$
let $p(i,j)\in\{+,-\}$ be the polarity with which variable $a_i$ occurs in clause $b_j$ (pick either one if both occur). This induces a two-coloring of the pairs of $\{1,\dots,h_u^p\}$. By Ramsey's theorem, there exists a subset of indices $I=\{i_1<\dots<i_m\}$ with $m \geq \frac{1}{2} \log_2 h_u^p$ and a color $c \in \{+,-\}$, such that $p(i_s,i_t)=c$ for all $s < t$. From this, we can build a signed semi-ladder $a'_1, \dots, a'_{m}, b_1', \dots, b_m'$ of height $m$ as follows: For each $t \in [m]$, let $a'_t$ be the literal of $a_{i_t}$ of polarity $c$ and let $b'_t=b_{i_t}$. This sequence is indeed a semi-ladder because (1) if $\{a_{i_t},b_{i_t}\}$ is a non-edge in the unsigned graph this means that neither literal of $a_{i_t}$ appears in the clause $b_{i_t}$, hence there is no edge between $a'_t$ and $b'_t$; and (2) for $s < t$, we have that $a'_s$ is adjacent to $b'_t$ because $a_{i_s}$ occurs in $b_{i_t}$ with polarity $c=p(a'_s)$.

We have thus shown that every proper unsigned semi-ladder of height $h_u^p$ yields a signed one of height $m \geq \frac{1}{2} \log_2 h_u^p$. Consequently, recalling that $\varphi$ has maximum signed semi-ladder height $h_s$ and that $h_u-2\leq h_u^p$, we can bound the maximum unsigned semi-ladder height by $h_u \leq h_u^p+2 \leq 4^{h_s}+2$.
\end{proof}

In Section~\ref{sec:lowerbounds}, we complement the tractability of finding weak backdoors on semi-ladder-free graphs by a lower bound that rules out a similar result for ladder-free graphs. There we will make use of the following notion.
\paragraph{Ladders.}
    A \emph{ladder of height $h$} in a graph $G$ is a sequence of vertices $a_1, \ldots, a_h, b_1, \ldots, b_h$ such that for all $i,j\in [h]$ we have $\{a_i,b_j\}\in E(G)$ if and only if $j \ge i$.
For $d\in \mathbb{N}$, we say that $G$ is \emph{$L_d$-free} if it contains no ladder of height at least $d$.

\paragraph{Backdoors.}
We now formally define the notion of weak backdoors for SAT.
\begin{definition}[Weak Backdoor to Treewidth $t$]
    Let $\varphi=(\var,\clause)$ be a CNF formula.
    A set~$X \subseteq \var$ is a \emph{weak backdoor of $\varphi$ to treewidth $t$} if there exists an assignment $\alpha\colon X \to \{0,1\}$ such that $\varphi[\alpha]$ is
    satisfiable and the unsigned incidence graph
    of $\varphi[\alpha]$ has treewidth at most $t$. Note that only satisfiable formulas admit weak backdoors.
    Since weak backdoors come with a specific assignment, we will also refer to them as sets of literals, where each variable in $X$ is set in the polarity that is satisfied by $\alpha$. The \emph{size} of a weak backdoor is the number of variables (equivalently, literals) it contains. A \emph{weak $(k,t)$-backdoor} is a weak backdoor of size at most $k$ to treewidth at most $t$.
\end{definition}

In Section~\ref{sec:main-tech-section}, we will mostly work on the following auxiliary problem, which allows us to make progress towards finding backdoors. We here state it in terms of a decision problem on colored graphs, See Figure~\ref{fig:ext_example}. Its last condition is based on the following notion, which measures, for each of $k$ given vertex sets, the largest proper semi-ladder supported by that set.

\begin{figure}[!htb]
    \centering
    \begin{tikzpicture}[
    every node/.style={font=\small},
    lit/.style={                 % literal vertex
        draw,
        circle,
        minimum size=0.6cm,
        thick,
        fill=violet!15,
        inner sep=0.5pt,
        font=\small
    },
    cls/.style={                 % clause vertex
        draw,
        rectangle,
        minimum size=0.6cm,
        thick,
        fill=teal!18,
        inner sep=2pt,
        font=\small
    },
    inc/.style={thick},                              % incidence edge
    opp/.style={ultra thick, dashed},             % opposite-literal edge
    leg/.style={font=\small, anchor=west}
]
    % ---- formula on the left ----
    \node[anchor=east, align=center, font=\large]
        (phi) at (0, -1.3)
        {$\varphi=(x_1 \vee \overline{x_2} \vee x_3)\wedge\,(\overline{x_2} \vee \overline{x_3})$};

    % ---- literal vertices in a row, grouped in opposing pairs ----
    \node[lit] (x1)  at (0.00,0) {$x_1$};
    \node[lit] (nx1) at (1,0) {$\overline{x_1}$};
    \node[lit] (x2)  at (2.3,0) {$x_2$};
    \node[lit] (nx2) at (3.3,0) {$\overline{x_2}$};
    \node[lit] (x3)  at (4.7,0) {$x_3$};
    \node[lit] (nx3) at (5.7,0) {$\overline{x_3}$};

    % ---- clause vertices, one above and one below the row ----
    \node[cls] (C1) at (2.6, -1.3) {$C_1$};
    \node[cls] (C2) at (5.0,-1.3) {$C_2$};

    % ---- opposite-literal edges (perfect matching) ----
    \draw[opp] (x1) -- (nx1);
    \draw[opp] (x2) -- (nx2);
    \draw[opp] (x3) -- (nx3);

    % ---- incidence edges ----
    \draw[inc] (C1) -- (x1);
    \draw[inc] (C1) -- (nx2);
    \draw[inc] (C1) -- (x3);
    \draw[inc] (C2) -- (nx2);
    \draw[inc] (C2) -- (nx3);

    % ---- legend ----
    \begin{scope}[shift={(7.2,0)}]
        \node[lit, scale=0.7] (lg1) at (0,0) {};
        \node[leg] at (0.4,0) {literal vertex};
        \node[cls, scale=0.7] (lg2) at (0,-0.55) {};
        \node[leg] at (0.4,-0.55) {clause vertex};
        \draw[inc] (-0.22,-1.1) -- (0.22,-1.1);
        \node[leg] at (0.4,-1.1) {incidence edge};
        \draw[opp] (-0.22,-1.65) -- (0.22,-1.65);
        \node[leg, text width=2.5cm] at (0.4,-1.65) {opposite-literal edge};
    \end{scope}
\end{tikzpicture}
    \caption{A CNF formula $\varphi$ and the graph $G$ in a corresponding instance $(k,s,t,G)$ of \bdextproblemtw. Initially, every literal vertex lies in each of the candidate sets $P_1,\dots,P_k$ (and holds the respective colors).}
    \label{fig:ext_example}
\end{figure}
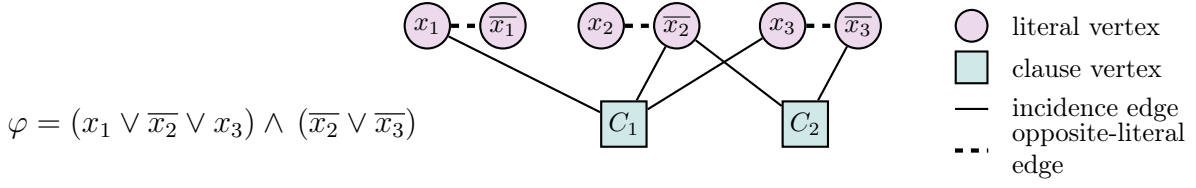

\begin{definition}[$SL$-pattern]
    Let $G$ be a graph in which every vertex is marked as a literal or a clause vertex, and let $P_1,\ldots, P_k \subseteq V(G)$. For each $i\in [k]$, let $s_i$ be the maximum height of a proper semi-ladder in $G$ whose literal vertices all lie in $P_i$. Then $(s_1, \ldots, s_k)$ is the \emph{$SL$-pattern} of $G$ with respect to $P_1, \ldots, P_k$, and we refer to $\sum_{i\in [k]} s_i$ as the \emph{$SL$-pattern size}.
\end{definition}

\begin{definition}[\backdoorextproblemtw\ (\bdextproblemtw)]\label{def:bdextproblemtw}
    For~$k,s,t\in\mathbb{N}$ and a vertex- and edge-colored graph $G$, the tuple $I=(k,s,t,G)$ is a \emph{well-formed instance} of \bdextproblemtw\ (and $G$ is \emph{well-formed}) if
    \begin{enumerate}
        \item[(W1)] every vertex of $G$ is colored as a \emph{literal vertex} or as a \emph{clause vertex}, but not both;
        \item[(W2)] every edge of $G$ is colored either as an \emph{incidence edge}, joining a literal vertex to a clause vertex, or as an \emph{opposite-literal edge}, joining two literal vertices;
        \item[(W3)] the opposite-literal edges form a perfect matching on the literal vertices---intuitively each pair represents the two opposing literals of some variable; 
        \item[(W4)] literal vertices may carry at most $k$ more colors indicating membership in some \emph{(backdoor) candidate sets} $P_1, \ldots, P_k$; 
        for every $i\in[k]$ the set $P_i$ is non-empty and $P_i\not\subseteq N_G(C)$ for every clause vertex $C$, that is, some literal vertex of $P_i$ is not an incidence-neighbor of $C$; and 
       \item[(W5)] the $SL$-pattern size of $G$ with respect to $P_1,\dots,P_k$ is at most $s$.
    \end{enumerate}
    
    A well-formed instance $I=(k,s,t,G)$ of {\bdextproblemtw} is a \emph{yes-instance} if there is an ordered set $X = \{\ell_1,\dots,\ell_k\}$ of literal vertices (a \emph{solution}) with $\ell_i\in P_i$ for each $i\in[k]$ that satisfies three properties:
    \begin{itemize}
     \item non-conflicting: there is no opposite-literal edge joining two solution vertices $\ell_i, \ell_j$;
     \item bounded treewidth: deleting $N_G[X]$ from $G$ and then contracting all opposite-literal edges yields a graph of treewidth at most $t$; and
     \item domination: there is a set of literal vertices in $G-N_G[X]$ that pairwise do not share any opposite-literal edges and which dominate all clause vertices in $G-N_G[X]$.
    \end{itemize}
    Non-well-formed instances are never yes-instances by any $X$, even if $X$ satisfies the above three criteria.
\end{definition}

The bounded treewidth property intuitively corresponds to having a set of literal vertices such that the residual \textsc{SAT} formula for this partial assignment has small treewidth. The domination property expresses that this partial assignment can be extended to a full satisfying assignment. The combination of these yields the property of being a small weak backdoor to bounded treewidth.

Crucially, we show that finding a weak $(k,t)$-backdoor for a CNF formula $\varphi$
is equivalent to solving a certain {\bdextproblemtw} instance where all sets $P_i$ contain all literal vertices.

The following observation relates the target treewidth $t$ in the \emph{bounded treewidth} property of yes-instances to the treewidth of $G-N_G[X]$ (without edge contraction).

\begin{observation}\label{obs:treewidth-translation}
    Let $X$ be a solution to a yes-instance $(k,s,t,G)$ of \bdextproblemtw. Then $G-N_G[X]$ and $G-N_G(X)$ have treewidth at most $2t+1$ each.
\end{observation}
\begin{proof}
    By definition, the contracted residual graph obtained by contracting edges of opposing literals in $G-N_G[X]$ has treewidth at most $t$. Consider a witness tree decomposition of width $t$ and replace in every bag each contracted vertex by the two endpoints of its matching edge. This is a valid tree decomposition of $G-N_G[X]$: every vertex and every edge is covered (a matching edge $\{v,\bar v\}$ by any bag holding its contracted vertex, and every other edge because the contracted vertices of its endpoints share a bag), and the bags containing any fixed vertex still form a subtree. 
    Previously, all bags had size at most $t+1$, so now they have size at most $2(t+1) = 2t+2$. Thus the treewidth of $G-N_G[X]$ is at most $2t+1$.
    Since $G-N_G(X)$ is obtained from $G-N_G[X]$ by adding the now isolated vertices of $X$, $G-N_G(X)$ has treewidth at most $2t+1$ as well (note that vertices in $X$ are pairwise non-adjacent as they are non-conflicting).
\end{proof}

\section{Weak Backdoors to Treewidth for Semi-Ladder-free SAT}\label{sec:main-tech-section}
The goal of this section is to establish our main algorithmic result Theorem~\ref{thm:main}: for every \textsc{SAT} instance whose signed incidence graph is $SL_d$-free, we can decide whether a weak $(k,t)$-backdoor exists---and, in the positive case, exploit it to solve \textsc{SAT}---in linear time.

At the core of our algorithm lies a recursive procedure formalized in the following theorem:

\begin{restatable}[Main Technical Theorem]{theorem}{mainTechnicalTheorem}\label{thm:mainTechnicalTheorem}
    There is an algorithm $\AA$ that takes as input a well-formed instance 
    $I=(k,s,t,G)$ of {\bdextproblemtw}, 
    where $k\ge 1$ and $G$ has $n$ vertices and $m$ edges. $\AA$ runs in time $\mathcal{O}_{k,s,t}(n+m)$ and outputs $q \le q(k,s,t)$ 
    many well-formed instances $I_1, \ldots, I_q$ of \bdextproblemtw\
    such that 
    \begin{itemize}
        \item $I_i = (k, s-1, t, G_i)$ or $I_i= (k-1, s, t, G_i)$ for some labeled graph $G_i$ for all $i\in [q]$; 
        \item $I$ is a yes-instance if and only if there is at least one $i\in [q]$ such that
        $I_i$ is a yes-instance; and 
        \item given a solution to $I_i$ for some $i \in [q]$, one can derive a solution to $I$ in $\mathcal{O}_{k,s,t}(n)$ time.
    \end{itemize}
\end{restatable}

Towards this, we first establish the three modules (reducer, extractor, and progressor) and combine them to prove Theorem~\ref{thm:mainTechnicalTheorem}, which is then used to establish our main Theorem~\ref{thm:main}.

\begin{restatable}[Reducer]{lemma}{reducerLemma}\label{lem:reducer}
    There exist $\rho = \rho(k,s,t) \in (0,1)$ and an algorithm $\RR$ which takes as input a well-formed instance $I = (k,s,t,G)$ of \bdextproblemtw, where $G$ has $n$ vertices and $m$ edges. $\RR$ runs in time $\mathcal{O}_{k,s,t}(n+m)$ and outputs a well-formed instance $I' = (k,s,t,G')$ of \bdextproblemtw\ such that
    \begin{itemize}
        \item $I$ is a yes-instance if and only if $I'$ is a yes-instance;
        \item given a solution to $I'$, one can compute a solution to $I$ in time $\mathcal{O}_{k,s,t}(n)$; and
        \item $I'$ is \emph{ready for extraction}: 
        for every solution $X$ to $I'$ (if there is any), $N_{G'}(X)$ is a $\rho$-cover in~$G'$.
    \end{itemize}
\end{restatable}

\begin{restatable}[Extractor]{lemma}{extractorLemma}\label{lem:extractor}
    There is an algorithm $\EE$ which takes as input a well-formed instance $I=(k,s,t,G)$ of {\bdextproblemtw} with $k\ge 1$, where $G$ has $n$ vertices and $m$ edges, and some $0<\rho<1$ such that for every solution $X$ to $I$, the set $N_G(X)$ is a $\rho$-cover in $G$.
    $\EE$ runs in time $(\frac{2k}{\rho})^{\mathcal{O}(t^2)} \cdot |G|$ and outputs a set $V'\subseteq V(G)$ such that
    \begin{itemize}
        \item if $I$ is a yes-instance witnessed by a solution $X$ 
        then $N_G[X] \cap V' \neq \emptyset$, and
        \item $|V'| \le (\frac{2k}{\rho})^{\mathcal{O}(t)}$.
    \end{itemize}
\end{restatable}

\begin{restatable}[Progressor]{lemma}{progressorLemma}\label{lem:branchingReduction}
    There is an algorithm $\PP$ which takes as input a well-formed instance $I=(k,s,t,G)$ of \bdextproblemtw, where $G$ has $n$ vertices and $m$ edges, and a vertex $v$ of $G$. $\PP$ runs in time $\mathcal{O}(k^2(n+m))$ and outputs a set $\mathcal{I}=\{I_1,\ldots, I_q\}$ of $q\le 2k$ well-formed instances of \bdextproblemtw such that
    \begin{itemize}
        \item $I_{i'}= (k-1, s, t, G_{i'})$ or $I_{i'}= (k, s-1, t, G_{i'})$ for some labeled graph $G_{i'}$ for all $i'\in [q]$;
        \item $I$ is a yes-instance via a solution $X$ with $v \in N_G[X]$ if and only if $\mathcal{I}$ contains a yes-instance;
        \item given a solution $X$ to an instance $I'\in \mathcal{I}$ one can derive a solution to $I$ in $\mathcal{O}(n)$ time.
    \end{itemize}
\end{restatable}

In the following sub-sections, we prove Lemmas~\ref{lem:reducer},~\ref{lem:extractor} and~\ref{lem:branchingReduction} and then use them to first establish Theorem~\ref{thm:mainTechnicalTheorem} and finally prove Theorem~\ref{thm:main}.

\subsection{The Reducer}\label{sec:reducer}

Our reducer (Lemma~\ref{lem:reducer}) turns a well-formed instance $I=(k,s,t,G)$ into an equivalent one whose every solution $X$ is ``sample-able'' in the sense that $N_G(X)$ covers a good fraction of all edge-endpoints of $G$; this is the $\rho$-cover property the extractor will consume in the next subsection. We obtain it by following the strategy of \cite[Theorem~3]{FominLMRS15_uib}\footnote{\cite{FominLMRS15_uib} is the openly available version of \cite{DBLP:conf/soda/FominLMRS15}.}:
A protrusion replacer is an algorithm to obtain an equivalent instance where a subgraph with a certain structure, a \emph{protrusion}, is replaced by a subgraph of fixed size. We will prove the existence of such an algorithm in Section~\ref{sec:protrusion-replacer}; for now, assume it is given as a black box.%

Using this black box algorithm, we describe a reducer in Section~\ref{sec:reducer-rho-cover}, which follows \cite{FominLMRS15_uib} closely. The key observation is that if many vertices of $G$ lie in large disjoint protrusions, we can replace them all in a single linear-time round (Lemma~\ref{lem:reducer}) and shrink the whole instance by a constant fraction (depending on $k,s,$ and $t$). We iterate this round until no such collection of protrusions remains. At that fixed point, a structural argument about $G$ (via Lemmas~6 and~7 in~\cite{FominLMRS15_uib}) shows that the neighborhood $N_G(X)$ of every weak backdoor $X$ must comprise a constant fraction of the vertices of $G$---intuitively, a small $N_G(X)$ would leave the bulk of $G$ packed into large protrusions, contradicting termination. A short degree-counting argument then converts this lower bound on $|N_G(X)|$ into the desired $\rho$-cover.

\subsubsection{The Protrusion Replacer}\label{sec:protrusion-replacer}

\paragraph{Boundaried Graphs and Protrusions.}
We use the standard boundaried-graph vocabulary for protrusion replacement; see~\cite[Section~2.2 and Definition~2.17]{DBLP:journals/corr/abs-0904-0727} and~\cite{DBLP:journals/jacm/BodlaenderFLPST16}\footnote{\cite{DBLP:journals/corr/abs-0904-0727} is the openly available version of \cite{DBLP:journals/jacm/BodlaenderFLPST16}.}. An \emph{$r$-boundaried graph} is a triple $\mathbf{G} = (G, \partial G, \lambda)$ where $G$ is a vertex- and edge-colored graph, $\partial G \subseteq V(G)$ has size at most $r$, and $\lambda\colon \partial G \to [r]$ is an injective \emph{labeling}. We lift standard graph notation to boundaried graphs as usual. In particular, let $\mathbf{G}[X] := (G[X], \partial G, \lambda)$ whenever $\partial G \subseteq X$. Two $r$-boundaried graphs are \emph{compatible} if their labelings induce an isomorphism between their boundary graphs $\mathbf{G}_1[\partial G_1]$ and $\mathbf{G}_2[\partial G_2]$ that respects vertex- and edge colors.
 The \emph{glue} $\mathbf{G}_1 \oplus \mathbf{G}_2$ of two $r$-boundaried compatible graphs $\mathbf{G}_1$ and $\mathbf{G}_2$ is the $r$-boundaried graph obtained from their disjoint union by identifying equally labeled boundary vertices and retaining that common labeled boundary. Compatibility makes the colors and edges among the identified vertices well-defined.
Forgetting the boundary of an $r$-boundaried graph $\mathbf{G}=(G,\partial G,\lambda)$ yields its \emph{underlying plain graph} $\textnormal{plain}(\mathbf{G}) := G$.

\begin{definition}[$r$-protrusion]\label{def:protrusion}
    Let $G$ be a well-formed vertex- and edge-colored graph and let $r \in \mathbb{N}$.
    A set $Y \subseteq V(G)$ is an \emph{$r$-protrusion} of $G$ if
    $|\partial_G(Y)| \le r$, where $\partial_G(Y) = \{v \in Y : (N_G(v) \setminus Y) \neq \emptyset\}$, and $\mathrm{tw}(G[Y]) \le r$. 
\end{definition}

\begin{definition}[Canonical equivalence]\label{def:canonical-equivalence}
    Fix \(k,s,t \in \mathbb{N}\).
    We say two boundaried graphs $\mathbf{G}_1, \mathbf{G}_2$ are \emph{canonically equivalent}
    (written $\mathbf{G}_1 \equiv_{k,s,t} \mathbf{G}_2$) if
    they are compatible 
    and,
    for every boundaried $\mathbf{H}$ compatible with both, $\textnormal{plain}(\mathbf{G}_1 \oplus \mathbf{H})$ is well-formed if and only if $\textnormal{plain}(\mathbf{G}_2 \oplus \mathbf{H})$ is, and whenever both are, it holds that
\begin{multline*}
        (k,s,t,\textnormal{plain}(\mathbf{G}_1 \oplus \mathbf{H})) \text{ is a yes-instance of } \bdextproblemtwbare
        \iff \\
        (k,s,t,\textnormal{plain}(\mathbf{G}_2 \oplus \mathbf{H})) \text{ is a yes-instance of } \bdextproblemtwbare.
\end{multline*}
\end{definition}
The relation \(\equiv_{k,s,t}\) is clearly an equivalence relation.
The following lemma shows that this equivalence relation has finite index in a constructive sense: every boundaried graph of bounded treewidth can be replaced, in linear time, by a bounded-size representative of its class, from which solutions can moreover be lifted back to the original.
\begin{lemma}[Finite index]\label{lem:finite-index}
    There is a computable function \(f\colon\mathbb{N}^4\to\mathbb{N}\) such that
    for every \(k,r,s,t \in \mathbb{N}\) and every \(r\)-boundaried graph \(\mathbf{G}\) of treewidth at most \(r\), one can compute in time \(f(k,r,s,t) \cdot |\mathbf{G}|\) 
    a representative \(\textnormal{rep}(\mathbf{G})\) with the following properties.
    \begin{enumerate}
        \item[(R1)] \emph{(Small and equivalent)} \(\textnormal{rep}(\mathbf{G})\) has size at most \(f(k,r,s,t)\) and satisfies \(\textnormal{rep}(\mathbf{G}) \equiv_{k,s,t} \mathbf{G}\).
        \item[(R2)] \emph{(Solution-respecting)} For every boundaried graph \(\mathbf{H}\) that is compatible with \(\mathbf{G}\) and every solution \(W\) of \((k,s,t,\textnormal{plain}(\textnormal{rep}(\mathbf{G})\oplus\mathbf{H}))\), one can compute in time \(f(k,r,s,t)\cdot|\mathbf{G}|\) a solution \(W'\) of \((k,s,t,\textnormal{plain}(\mathbf{G}\oplus\mathbf{H}))\).
    \end{enumerate}
\end{lemma}
\begin{proof}
    Throughout, we work over the fixed finite signature consisting of the two vertex colors (\emph{literal} and \emph{clause}), the unary predicates \(P_1,\dots,P_k\), the two edge predicates (\emph{incidence} and \emph{opposite-literal}), and \(2r+1\) unary predicates for boundary labels. An \(r\)-boundaried graph uses only the first \(r\) of these predicates; the remaining predicates are used for the bags of the tree decomposition below.

    The only problem-specific ingredients are two $\mathrm{MSO}_2$ sentences \(\xi\) and \(\phi\). For an introduction to $\mathrm{MSO}_2$, we refer to \cite{Cou90} and \cite[Chapter~7.4.1]{CyganFKLMPPS15}.
    We claim that, for every plain graph \(J\),
    \begin{align*}
        J\models\xi &\iff J\text{ is well-formed},\\
        J\models\phi &\iff J\text{ is a well-formed yes-instance of }\bdextproblemtwbare.
    \end{align*}
    To define \(\xi\), observe that in Definition~\ref{def:bdextproblemtw}, conditions (W1) and (W2) are first-order constraints on the colors\footnote{We here simply interpret the first color as the color for literal vertices, the second for clause vertices, and so on, in an arbitrary fixed order.} and edges, (W3) says that every literal vertex is incident to exactly one opposite-literal edge, and (W4) essentially says that each \(P_i\) is non-empty and that every clause vertex misses some \(P_i\)-vertex. For~(W5), a proper semi-ladder of any fixed height \(h\) is a fixed finite configuration. Hence, the statement that there is a semi-ladder that has its entire literal side in \(P_i\) is first-order expressible.  Since \(s\) and \(k\) are fixed, the condition \(\sum_{i\in[k]}s_i\le s\) for the \(SL\)-pattern \((s_1,\dots,s_k)\) can be expressed by forbidding, for every tuple of non-negative integers \((h_1,\dots,h_k)\) with \(\sum_i h_i=s+1\), the simultaneous existence of proper semi-ladders of heights \(h_1,\dots,h_k\) whose literal sides lie in \(P_1,\dots,P_k\), respectively; a height-zero requirement is vacuous. Thus, \(\xi\) exists.

    To define \(\phi\), let \(\psi(X)\) be the $\mathrm{MSO}_2$ formula with one free vertex-set variable \(X\) that asserts that \(X\) consists of a valid selection of non-conflicting literals from the sets \(P_i\), and states that the contracted residual graph $(J-N_J[X])^c$, obtained by contracting all opposite-literal edges in $J-N_J[X]$,  has treewidth at most \(t\) and is satisfiable. For fixed \(t\), the property ``treewidth at most \(t\)'' is $\mathrm{MSO}_2$-expressible and satisfiability is expressed by the existence of a set \(S\) of literal vertices that contains no opposite-literal edge and in which every clause vertex has an incidence-neighbor. Moreover, the contracted residual graph is $\mathrm{MSO}_2$-interpretable from \(X\).
    \footnote{For fixed \(t\), the graphs of treewidth at most \(t\) form a minor-closed class and are therefore characterized by a finite set of forbidden minors; since ``\(H\) is a minor'' is $\mathrm{MSO}_2$-expressible for each fixed graph \(H\), so is ``treewidth at most \(t\)'' (see, e.g.,~\cite[Chapter~7]{CyganFKLMPPS15}).} We then set
    \[
        \phi:=\xi \wedge \exists X\,\psi(X).
    \]

    With \(\xi\) and \(\phi\) in hand, the lemma readily follows from the finite-state and protrusion-replacement machinery for $\mathrm{CMSO}$-definable properties of boundaried graphs developed by Bodlaender, Fomin, Lokshtanov, Penninkx, Saurabh, and Thilikos~\cite{DBLP:journals/corr/abs-0904-0727}. Since $\bdextproblemtwbare$ is a decision problem, we need only the finite-state property of that framework; the finite-integer-index refinement used there for optimization problems is not required.

    For completeness, we briefly summarize the key ideas of~\cite{DBLP:journals/corr/abs-0904-0727}.
    Let \(q\) be the maximum of the quantifier ranks of \(\xi\) and \(\phi\). The \emph{\(q\)-type} of a graph \(\mathbf{G}\), denoted \(q\textnormal{-type}(\mathbf{G})\), is the set of $\mathrm{MSO}_2$ sentences of quantifier rank at most \(q\) that hold in it. Up to logical equivalence, there are only finitely many such sentences over the fixed signature. Hence the number of \(q\)-types is bounded by a computable function of \(k,r,s,t\).
    We rely on a Feferman--Vaught-style compositionality theorem for \(\mathrm{MSO}_2\) on boundaried graphs: the \(q\)-type of a glue \(\mathbf{A}\oplus\mathbf{B}\) is determined by the \(q\)-types of \(\mathbf{A}\) and \(\mathbf{B}\)~\cite{FefermanV59,Mak04}.
    We apply it three times.

    First, every boundary vertex has a unique unary predicate, and w.l.o.g.\ \(q\ge 2\). Therefore, the \(q\)-type determines the labeled boundary graph. Thus, two boundaried graphs of equal \(q\)-type are compatible. When they are glued to any compatible context \(\mathbf{H}\), compositionality yields glues of equal \(q\)-type, whose plain graphs therefore agree on \(\xi\) and \(\phi\). Hence equal \(q\)-type implies canonical equivalence. In particular, the number of \(q\)-types bounds the number of canonical equivalence classes.

    Second, take a rooted tree decomposition for $G-\partial G$ of width at most \(r\), add \(\partial G\) to every bag, and convert the result into a nice tree decomposition for $G$. The resulting decomposition has width at most \(2r\), and its root contains \(\partial G\). For each bag, consider the boundaried graph consisting of the bag together with all vertices introduced in the subtree rooted at that bag, using the bag as the boundary. Its \(q\)-type can be computed bottom-up using Courcelle-style arguments.
    Whenever the same \(q\)-type occurs twice on a root-to-leaf path, replace the larger rooted subgraph by the smaller one having the same $q$-type. Compositionality ensures that this replacement leaves the \(q\)-type of the whole graph unchanged. Since the original boundary \(\partial G\) is retained throughout the construction, the resulting representative is compatible with \(\mathbf{G}\).
    Repeating this operation yields a graph of the same \(q\)-type whose size (number of edges plus number of vertices) is bounded by a computable function of \(k,r,s,t\): a nice tree decomposition has bounded branching, and after the replacements its depth is bounded by the number of possible \(q\)-types.
    Computing the bag \(q\)-types and performing the replacements takes time \(f(k,r,s,t)\cdot|\mathbf{G}|\) by Courcelle's Theorem~\cite{Cou90}. Together with the first application, this proves (R1).

    For the third application, we transfer witnesses to establish (R2). Since \(\exists X\,\psi(X)\) is a subformula of~\(\phi\), the quantifier rank of \(\psi\) is at most \(\mathrm{qr}(\phi)-1\le q-1\).
    For a graph \(\mathbf{G}\) and a set \(S \subseteq V(\mathbf{G})\), denote by \(q\textnormal{-type}(\mathbf{G},S)\) the \(q\)-type of the structure obtained by expanding \(\mathbf{G}\) with a fresh unary relation interpreted as \(S\). The Feferman--Vaught-style compositionality theorem extends to this enriched signature~\cite{FefermanV59,Mak04}: for compatible boundaried graphs \(\mathbf{A},\mathbf{B}\) and any \(U \subseteq V(\mathbf{A}\oplus\mathbf{B})\), the labeled type \(q\textnormal{-type}(\mathbf{A}\oplus\mathbf{B},U)\) is determined by \(q\textnormal{-type}(\mathbf{A},U\cap V(\mathbf{A}))\) and \(q\textnormal{-type}(\mathbf{B},U\cap V(\mathbf{B}))\).

    Now let \(W \subseteq V(\textnormal{plain}(\textnormal{rep}(\mathbf{G})\oplus\mathbf{H}))\) be a solution, and set \(W_R := W \cap V(\textnormal{rep}(\mathbf{G}))\) and \(W_H := W \cap V(\mathbf{H})\). Let \(\tau := (q{-}1)\textnormal{-type}(\textnormal{rep}(\mathbf{G}), W_R)\).
    Let \(\rho\) be the finite set of (normalized) $\mathrm{MSO}_2$ sentences of quantifier rank at most \(q-1\) over the extended signature, and define
    \[
        \psi_\tau(X) := \bigwedge_{\sigma \in \rho \cap \tau}\sigma \;\wedge\; \bigwedge_{\sigma \in \rho \setminus \tau}\neg\sigma,
    \]
    where the additional unary relation is viewed as the free set variable \(X\). By construction, \(\mathbf{G} \models \psi_\tau(S)\) if and only if \((q{-}1)\textnormal{-type}(\mathbf{G},S) = \tau\), and the realizability sentence \(\exists X\,\psi_\tau(X)\) has quantifier rank at most~\(q\). Since \(\textnormal{rep}(\mathbf{G}) \models \exists X\,\psi_\tau(X)\), witnessed by \(W_R\), and \(\textnormal{rep}(\mathbf{G})\) and \(\mathbf{G}\) have the same \(q\)-type, we also have \(\mathbf{G} \models \exists X\,\psi_\tau(X)\). The constructive variant of Courcelle's Theorem due to Arnborg, Lagergren, and Seese~\cite{ArnborgLS91}, applied to the boundaried graph \(\mathbf{G}\) of treewidth at most \(r\), therefore produces a set \(W_G \subseteq V(\mathbf{G})\) with \((q{-}1)\textnormal{-type}(\mathbf{G},W_G)=\tau\) in time \(f(k,r,s,t)\cdot|\mathbf{G}|\).

    Because each boundary label has its own unary predicate, the type \(\tau\) determines which boundary vertices belong to the distinguished set. Hence \(W_G\) and \(W_R\) agree on equally labeled boundary vertices, and \(W_G\) can be combined consistently with \(W_H\). Define \(W' := W_G \cup W_H \subseteq V(\mathbf{G}\oplus\mathbf{H})\). By labeled compositionality, we have that
    \[
        (q{-}1)\textnormal{-type}(\mathbf{G}\oplus\mathbf{H}, W') = (q{-}1)\textnormal{-type}(\textnormal{rep}(\mathbf{G})\oplus\mathbf{H}, W).
    \]
    The right-hand side satisfies \(\psi\), which has rank at most \(q-1\), because \(W\) is a solution. Therefore, the left-hand side also satisfies \(\psi\). Thus, \(W'\) is a solution, proving (R2).
\end{proof}

Replacing a single protrusion by its representative is a local, in-place modification of the graph; we record it in the following lemma.

\begin{lemma}[Protrusion replacement]\label{lem:protrusionReplacer}
    There is an algorithm that takes as input an instance $I=(k,s,t,G)$ of {\bdextproblemtwbare} and an $r$-protrusion $Y$ of $G$ with $|Y|\ge f(k,r,s,t)+1\eqcolon r'$, where $f$ is the function from Lemma~\ref{lem:finite-index}.
   In $\mathcal{O}(r'\cdot (r^2 + r|Y|))$ time, the algorithm modifies $I$ in place into an instance $I'=(k,s,t,G')$ with $|V(G')|<|V(G)|$ such that $G'$ is well-formed if and only if $G$ is and $I$ is a yes-instance if and only if $I'$ is. Moreover, given any solution of $I'$, a solution of $I$ can be computed in $\mathcal{O}(r'\cdot(r^2 + r|Y|))$ time.
\end{lemma}
\begin{proof}
    Recall that we assumed that our graph encoding supports adjacency tests as well as the addition (deletion, respectively) of an edge (vertex, resp.) in $\mathcal{O}(1)$ time each, as well as identifying $N_G(v)$ from a vertex $v$ in $\mathcal{O}(|N_G(v)|)$ time.
    We first assemble the boundaried graph $\mathbf{G}_Y=(G[Y],\partial_G(Y),\lambda)$, where $\lambda\colon \partial_G(Y)\to[r]$ is an arbitrary injective labeling. Since $Y$ is an $r$-protrusion, $|\partial_G(Y)|\le r$ and $G[Y]$ has at most $\mathcal{O}(r^2)$ edges in the boundary and $\mathcal{O}(r |Y|)$ edges incident to at least one non-boundary vertex. Since non-boundary vertices have no neighbors outside $Y$ in $G$, the latter edges can be identified in $\mathcal{O}(r|Y|)$ time. The edges between boundary vertices are found by $\mathcal{O}(r^2)$ adjacency tests. Hence $\mathbf{G}_Y$ is assembled in $\mathcal{O}(r^2 + r|Y|)$ time.

    By (R1) of Lemma~\ref{lem:finite-index} we compute the representative $\textnormal{rep}(\mathbf{G}_Y)$, which satisfies $\mathbf{G}_Y\equiv_{k,s,t}\textnormal{rep}(\mathbf{G}_Y)$ and $|\textnormal{rep}(\mathbf{G}_Y)| < r'$, in time $r'\cdot|\mathbf{G}_Y|=\mathcal{O}(r'\cdot (r^2 + r|Y|))$. We then update $G$ to $G'$ in place: delete the interior $Y\setminus \partial_G(Y)$ together with its incident edges (in $\mathcal{O}(r|Y|)$ time), and insert the interior of $\textnormal{rep}(\mathbf{G}_Y)$, attaching it to the retained boundary $\partial_G(Y)$ (in $\mathcal{O}(r')$ time). 
    
    The size guarantee and solution lifting follow immediately from Lemma~\ref{lem:finite-index}, using $|\mathbf{G}_Y| \le (r^2 + r|Y|)$. The preservation of well-formedness as well as outputting a yes-instance if and only if $I$ is a yes-instance is due to the canonical equivalence.
\end{proof}

\subsubsection{Constructing the Reducer}\label{sec:reducer-rho-cover}

While Lemma~\ref{lem:protrusionReplacer} handles replacing a single protrusion, our reducer module (Lemma~\ref{lem:reducer}) replaces many protrusions per round, using the lemma as a black box. Such a round shrinks the instance size when the current graph admits a certain structure, namely an $(\alpha, \beta)$-protrusion decomposition as used by~\cite{FominLMRS15_uib} and originally introduced by~\cite{DBLP:journals/corr/abs-0904-0727}.
\begin{definition}[$(\alpha,\beta)$-protrusion decomposition]\label{def:protrusion-decomposition}
    Let $\alpha,\beta\in\mathbb{N}$. An \emph{$(\alpha,\beta)$-protrusion decomposition} of a graph $G$ is a partition $V(G) = R_0 \uplus R_1 \uplus \dots \uplus R_q$ of its vertex set such that
    \begin{itemize}
        \item $|R_0| + q \le \alpha$;
        \item for every $i \in [q]$, $N_G(R_i) \subseteq R_0$ and $|N_G(R_i)| \le \beta$;
        \item for every $i \in [q]$, $R_i \cup N_G(R_i)$ is a $\beta$-protrusion of $G$ with boundary $N_G(R_i)$.
    \end{itemize}
\end{definition}
Intuitively, such a decomposition describes a set of disjoint $\beta$-protrusions in the graph, where $\alpha$ bounds the number of protrusions and the size of the graph not covered by protrusions. 

A small $\alpha$ relative to $|V(G)|$ thus means that the bulk of~$G$ is packed into protrusions, which is precisely the regime in which the protrusion replacer of Lemma~\ref{lem:protrusionReplacer} can be applied to a large fraction of $G$ in a single linear-time round.

We now show how multiple protrusions can be replaced in a single linear-time round. This is the Linear Time Protrusion Replacement Theorem (Theorem~3 in~\cite{FominLMRS15_uib}), instantiated for {\bdextproblemtw}; where we tighten the running time guarantee to $\mathcal{O}(n)$ rather than the $\mathcal{O}(n+m)$ stated by~\cite{FominLMRS15_uib}. This improvement is what makes the iterated reducer linear: across the rounds the number of vertices shrinks by a constant factor but the number of edges need not, so an $\mathcal{O}(n+m)$ per-round cost would not sum to linear (cf.\ the running-time analysis in the proof of Lemma~\ref{lem:reducer}); for \cite{FominLMRS15_uib} this did not matter, since there $m\in \mathcal{O}(n)$.

\begin{lemma}[One-round protrusion replacement; cf.\ {\cite[Theorem~3]{FominLMRS15_uib}}]\label{lem:linear-protrusion-replacement}
     There is an algorithm that takes as input an instance $I=(k,s,t,G)$ of {\bdextproblemtwbare}, where $G$ has $n$ vertices.
   In $\mathcal{O}_{k,s,t}(n)$ time, the algorithm computes an instance $I'=(k,s,t,G')$ in place such that $|V(G')|\le n$, $G'$ is well-formed if and only if $G$ is and $I'$ is a yes-instance if and only if $I$ is.
    Moreover, given a solution to $I'$, a solution to~$I$ can be recovered in $\mathcal{O}_{k,s,t}(n)$ time.
     If, in addition, $G$ admits an $(\alpha,\beta)$-protrusion decomposition with $\beta\le 4t+4$ and $\alpha\le\frac{n}{244(r'\cdot 2^r)}$ for $r=12t+15$ and $r'$ as in Lemma~\ref{lem:protrusionReplacer}, then $|V(G')|\le(1-\delta)\,n$ for some $\delta=\delta(k,s,t)\in(0,1)$.
\end{lemma}
\begin{proof}
This is the algorithm of Theorem~3 of~\cite{FominLMRS15_uib} applied to our problem using the protrusion replacer of Lemma~\ref{lem:protrusionReplacer}---compute a large protrusion cover with the linear-time algorithm of \cite[Lemma~16]{FominLMRS15_uib} and replace each of its protrusions. 
Since each single replacement preserves well-formedness (Lemma~\ref{lem:protrusionReplacer}), so does their composition; hence $G'$ is well-formed if and only if $G$ is.
Below, we comment on some changes we made to that statement apart from applying it to our problem: we fixed some constants, added the recovery of solutions, and decreased the running time to be linear in $n$ instead of $n+m$.

    \paragraph*{Setting Constants.} We instantiate the algorithm of~\cite{FominLMRS15_uib} with $r=12t+15$, and $r'$ as in Lemma~\ref{lem:protrusionReplacer}, and we let their variable~$s$---not to be confused with our $s$---take value $r'\cdot 2^r$. 
    Their requirement $r \ge 3(\beta+1)$ is equivalently expressed by $\beta \le 4t+4$.

    \paragraph*{Reconstruction of solutions.} 
    While this is not explicitly stated in \cite{FominLMRS15_uib}, the authors rely on solution recovery as well. It follows by observing that the protrusions can be replaced one after another, using Lemma~\ref{lem:protrusionReplacer}, in any order. Indeed, replacing a protrusion $Y$ deletes only the interior of $Y$---these vertices are exclusive to $Y$, since the protrusions computed by the algorithm overlap only on their boundaries, if at all---and it retains $Y$'s labeled boundary together with the edges among boundary vertices. Hence, for any other protrusion~$\hat Y$, no vertex of $\hat Y$ is deleted and the subgraph induced on $\hat Y$ is unchanged, so $\hat Y$ remains an $r$-protrusion of size at least $r'$ and Lemma~\ref{lem:protrusionReplacer} still applies to it. Tracing back the resulting solutions in a last-in-first-out manner and at each step invoking the solution-recovery property of Lemma~\ref{lem:protrusionReplacer} allows us to derive a solution to $I$ from a solution to $I'$. We replace a set $\mathcal{Y}$ of protrusions and the recovery step takes time $\mathcal{O}(r'(r^2+r|Y|))$ for each $Y\in \mathcal{Y}$. As the protrusions pairwise overlap in at most $r$ vertices (the boundary), we have 
    \[\sum_{Y\in \mathcal{Y}} r'(r^2+r|Y|) \le r'r^2n + r'r\cdot \sum_{Y\in \mathcal{Y}} (r+1) |Y'|  = \mathcal{O}(r'\, r^2\,n ),\]
     where $Y' \subseteq Y$ is the set of vertices which appear exclusively in $Y$.

    \paragraph*{Tighter time bound.} The authors of \cite{FominLMRS15_uib} bound the running time of their algorithm by $\mathcal{O}_{r',r}(n+m)$. We analyze their algorithm more carefully, tightening the running time to $\mathcal{O}_{r',r}(n) = \mathcal{O}_{k,s,t}(n)$.
    Note that \cite[Lemma~16]{FominLMRS15_uib} returns a set of $\mathcal{O}(n)$ protrusions that only overlap in the boundaries and which are then sequentially replaced using Lemma~\ref{lem:protrusionReplacer}, with each protrusion $Y$ taking $\mathcal{O}(r'\cdot(r^2 + r|Y|))$ time. Since there are $\mathcal{O}(n)$ protrusions and, by the same argument as above, $\sum_{Y\in\mathcal{Y}}|Y| \le (r+1)n$, this latter step of replacing all protrusions takes $\mathcal{O}(r^2\cdot r' \cdot n)$ time.
    It remains to examine the time required by \cite[Lemma~16]{FominLMRS15_uib} more closely. We do not restate the arguments of correctness or even the intuition of this proof, but simply list their steps and argue that each is feasible in $\mathcal{O}_{r,r'}(n)$ time.
    \begin{enumerate}
        \item For some constant $s'$ depending only on $r$ and $r'$, the authors of~\cite{FominLMRS15_uib} first compute for every vertex $v$ a set $\mathcal{S}_v$ that contains all $S\subseteq V(G)$ such that we have $v\in S, |S| \le 6s'$, $|N_G(S)|\le r$, and $G[S]$ is connected. They explicitly state this takes $\mathcal{O}_{r,r'}(n)$ time in total (independently on whether $m\in \mathcal{O}_{r,r'}(n)$ or larger). This is due to the fact that no such set $S$ contains vertices of degree larger than $6s'+r$, hence, all high degree vertices can be ignored. Identifying whether the degree of a vertex exceeds some constant is possible in constant time. They remark that $|\mathcal{S}_v|=\mathcal{O}_{r,r'}(1)$ for every vertex $v$.
        \item Next, for every vertex $v$, all sets $S\in \mathcal{S}_v$ for which the treewidth exceeds some constant value are discarded. As is argued in~\cite{FominLMRS15_uib}, this can be tested in $\mathcal{O}_{r,r'}(1)$ time per set and thus takes $\mathcal{O}_{r,r'}(n)$ time in total.
        \item As a next step, the authors group all sets in $\bigcup_{v\in V(G)} \mathcal{S}_v$ according to their neighborhood, which can be done in $\mathcal{O}_{r,r'}(n)$ time as they argue themselves. This can be achieved by $r$ passes of bucket sort over the vertex identifiers; to keep the cost at $\mathcal{O}_{r,r'}(n)$ even when these identifiers do not form the range $[n]$ (as happens when the graph has been locally modified previously), we clear the buckets lazily, touching only the $\mathcal{O}_{r,r'}(n)$ buckets actually used rather than reinitializing the whole array. This yields a partition $\mathcal{X}_1, \ldots, \mathcal{X}_q$ into non-empty sets, so $q \le \sum_{v\in V(G)} |\mathcal{S}_v| \le n \cdot \mathcal{O}_{r,r'}(1) = \mathcal{O}_{r,r'}(n)$.
        \item Then they form sets $\mathcal{Y}_1, \ldots, \mathcal{Y}_q$ by, for each $\mathcal{X}_i$, greedily extracting disjoint subsets of certain sizes. Since each set $S$ is touched at most once, this takes $\mathcal{O}_{r,r'}(q + n) = \mathcal{O}_{r,r'}(n)$ time.
        \item From this, they construct a graph $H$ with vertex set $\{\mathcal{Y}_1, \ldots, \mathcal{Y}_q\}$, where $\mathcal{Y}_i$ and $\mathcal{Y}_j$ with $i\neq j$ share an edge if and only if there are $Y_i\in \mathcal{Y}_i$ and $Y_j\in \mathcal{Y}_j$ such that $Y_i\cap Y_j \neq \emptyset$ or there is an edge in $G$ between $Y_i$ and $Y_j$. 
        Assume $Y_i\cap Y_j = \emptyset$, then the only case where there is an edge between $Y_i$ and $Y_j$ in $G$ is if one set intersects the neighborhood of the other, $N_G(Y_i)\cap Y_j \neq \emptyset$ or conversely $N_G(Y_j)\cap Y_i \neq \emptyset$. Thus, $\mathcal{Y}_i$ and $\mathcal{Y}_j$ share an edge if and only if there are $Y_i\in \mathcal{Y}_i$ and $Y_j\in \mathcal{Y}_j$ such that $N_G[Y_i]\cap Y_j \neq \emptyset$.
        By construction (the third and fourth step), each collection $\mathcal{Y}_i$ is a sub-collection of one of the neighborhood classes of $\mathcal{S}\coloneq\bigcup_{v\in V(G)}\mathcal{S}_v$, so every member $Y_i\in\mathcal{Y}_i$ is one of the sets listed in the first step. Moreover, the only condition defining $\mathcal{S}_v$ that refers to $v$ is $v\in S$, while the remaining conditions ($|S|\le 6s'$, $|N_G(S)|\le r$, and $G[S]$ connected) depend on $S$ alone; hence $Y_i\in\mathcal{S}_v$ for \emph{every} $v\in Y_i$, and likewise $Y_j\in\mathcal{S}_w$ for every $w\in Y_j$. We exploit this to construct $H$ as follows. Iterate over every vertex $v$ and every set $S\in\mathcal{S}_v$; since $|N_G[S]|\le 6s'+r$ and $|\mathcal{S}_w|=\mathcal{O}_{r,r'}(1)$ for every vertex $w$, we can in time $(6s'+r)\cdot\mathcal{O}_{r,r'}(1)=\mathcal{O}_{r,r'}(1)$ iterate over all vertices $w\in N_G[S]$ and all sets $S'\in\mathcal{S}_w$. Whenever $S$ and $S'$ lie in distinct collections $\mathcal{Y}_i$ and $\mathcal{Y}_j$, add the edge $\{\mathcal{Y}_i,\mathcal{Y}_j\}$ to $H$ (unless it is already present).
        This enumeration finds every edge of $H$: if $\mathcal{Y}_i$ and $\mathcal{Y}_j$ are adjacent, witnessed by $Y_i\in\mathcal{Y}_i$, $Y_j\in\mathcal{Y}_j$, and a vertex $w\in N_G[Y_i]\cap Y_j$, then choosing any $v\in Y_i$ the outer loop enumerates $S=Y_i\in\mathcal{S}_v$, the inner loop reaches $w\in N_G[S]$, and there it enumerates $S'=Y_j\in\mathcal{S}_w$ (using $w\in Y_j$). Conversely, every pair produced this way satisfies $N_G[S]\cap S'\neq\emptyset$ and hence corresponds to a genuine edge. As there are $\mathcal{O}_{r,r'}(n)$ choices of~$v$, each handled in $\mathcal{O}_{r,r'}(1)$ time, $H$ can be constructed in $\mathcal{O}_{r,r'}(n)$ time.
        \item Lastly, they compute a large independent set in $H$ in $\mathcal{O}_{r,r'}(|V(H)|+|E(H)|) = \mathcal{O}_{r,r'}(q) = \mathcal{O}_{r,r'}(n)$ time, noting that $H$ has constant maximum degree for fixed $r,r'$. The output of \cite[Lemma~16]{FominLMRS15_uib} are the sets corresponding to this independent set in $H$.
    \end{enumerate}
    Since each of the six consecutive steps takes $\mathcal{O}_{r,r'}(n)$ time, so does the entire algorithm.  
\end{proof}

With Lemma~\ref{lem:linear-protrusion-replacement} in hand, we are now ready to prove the main lemma of this subsection.
\reducerLemma*
\begin{proof}%
    Use $r=12t+15$, $r'$, $b=r'\cdot 2^r$, and $\delta$ of Lemma~\ref{lem:linear-protrusion-replacement}. Let $\QQ$ denote a single \emph{round}, that is, one application of Lemma~\ref{lem:linear-protrusion-replacement} to the current graph $G_i$ on $n_i=|V(G_i)|$ vertices, which in time $\mathcal{O}_{k,s,t}(n_i)$---independent of the number of edges---produces an equivalent graph $G_{i+1}$ with $n_{i+1} \coloneq |V(G_{i+1})|\le n_i$ vertices. By Lemma~\ref{lem:linear-protrusion-replacement}, if $G_i$ admits an $(\alpha,\beta)$-protrusion decomposition with $\beta\le 4t+4$ and $\alpha\le n_i/244b$, then $n_{i+1} \le(1-\delta)n_i$.

    Set $G_0 \coloneqq G$ and iterate rounds $G_{i+1} \coloneqq \QQ(G_i)$. Stop at the first index $i^*$ at which the round failed to shrink, i.e., $n_{i^*+1} > (1-\delta)\,n_{i^*}$, and let $G' \coloneqq G_{i^*}$ and $I' \coloneqq (k,s,t,G')$.
    The total time required is $\sum_{i\le i^*} \mathcal{O}_{k,s,t}(n_i)$. Since $n_{i+1} \le(1-\delta)n_i$ for all such $i$ (except for $i^*$), this forms a geometric series bounded by $\mathcal{O}_{k,s,t}(n_0)=\mathcal{O}_{k,s,t}(n)$.

    By induction along the chain, $I'$ is a yes-instance if and only if $I$ is, and---by composing the per-round solution recovery of Lemma~\ref{lem:linear-protrusion-replacement}---a solution of $I'$ lifts to a solution of $I$ in time $\sum_{i<i^*}\mathcal{O}_{k,s,t}(n_i)=\mathcal{O}_{k,s,t}(n)$. Each round also preserves well-formedness (Lemma~\ref{lem:linear-protrusion-replacement}), so since $I$ is well-formed, so is~$I'$. This establishes that $I'$ is well-formed together with the first two statements of the lemma.

    It remains to argue that the third condition is met as well. 
    For this, let $n' \coloneqq |V(G')|$ and let $X$ be a solution to~$I'$.
    By Observation~\ref{obs:treewidth-translation}, $G'-N_{G'}[X]$ has treewidth at most $2t+1$. 
    Then, feeding $N_{G'}(X)$, $\eta = 2t+1$, and $a= (8t+8)r' \cdot 2^r > 1$ to Lemma~6 in \cite{FominLMRS15_uib} (which is based on \cite{DBLP:journals/corr/abs-0904-0727,FominLMS12}) gives that at least one of the following holds:
    \[
        \text{(i)}\phantom{-} |N_{G'}[N_{G'}(X)]| \le n'/(244a) \quad\text{ or }\quad  
    \text{(ii)}\phantom{-} N_{G'}(X) \text{ is a } \rho\text{-cover in } G' \text{ for }\rho = \frac{1}{488a(2t+2)+2}.
    \]

    Towards a contradiction, assume that (i) is the case.
    Feeding $N_{G'}(X)$ to Lemma~7 in \cite{FominLMRS15_uib} (which is based on \cite{DBLP:journals/corr/abs-0904-0727}) yields an $(\alpha, \beta)$-protrusion decomposition with $\alpha= |N_{G'}[N_{G'}(X)]|\cdot 4(2t+2)$ and $\beta = 2(2t+2)$.
    Since $\alpha \le n'/(244a) \cdot (8t+8) = n'/(244b)$ and $\beta = 4t+4$, the shrinking condition of Lemma~\ref{lem:linear-protrusion-replacement} would be satisfied, so a further round $\QQ$ applied to $G'$ would have shrunk it---contradicting the stopping condition. Consequently, case (i) never occurs and $N_{G'}(X)$ is a $\rho$-cover.
\end{proof}

\subsection{The Extractor}

The extractor is the second module that Theorem~\ref{thm:mainTechnicalTheorem} relies on. It receives an instance $I$ of \bdextproblemtw in which the neighborhood of every solution is a $\rho$-cover of the graph associated with $I$
and aims at identifying a set of vertices such that, for any solution to $I$, at least one of these vertices lies in the solution or in its neighborhood.

For a single fixed solution \(X\), the hitting property is easy to achieve by randomization: since \(N_G(X)\) is a \(\rho\)-cover, a uniformly random edge has an endpoint in \(N_G(X)\) with probability at least~\(\rho\), so collecting the endpoints of \(\mathcal{O}(\rho^{-1})\) random edges hits \(N_G[X]\) with constant probability. The more elaborate lemma below instead produces a single deterministic set that works for all solutions at once.

\extractorLemma*
\begin{proof}
Recall that $V(G)$ is partitioned into a set $\mathcal{C}$ of clause vertices and a set $\mathcal{L}$ of literal vertices.
   Set $\alpha:=\frac{\rho}{k}$, $\beta:=\frac{\rho}{2(2t+1)k}$, and $\gamma := \prod_{i=0}^{2t+1} (\alpha-i\beta)$.
   Note that $0 < \gamma < 1$, since each factor lies in $(0,1)$.
   We define the following sets of vertices.

   $$W:= \{\ell \in \mathcal{L}:  \deg_{G}(\ell) \geq \alpha m\},\qquad Y:= \{C \in  \mathcal{C}:  \deg_{G}(C) \geq \beta m\},$$
   $$ Z:=\{v \in \mathcal{L} \setminus N_G(Y): \sum_{C \in N_{G}(v)\cap\mathcal{C}}\deg_{G}(C) \geq \alpha m\}$$

   We first show that $V' = W \cup Y \cup Z$ is always a valid output apart from the size bound.

  \begin{claim}\label{claim:sampling-one}
      If $X$ is a solution to $I$, then
      $N_{G}[X] \cap (W\cup Y \cup Z) \neq \emptyset$.
  \end{claim}
  \begin{claimproof}[Proof of Claim~\ref{claim:sampling-one}]
      Towards a contradiction, suppose that $N_G[X]\cap(W\cup Y\cup Z)=\emptyset$; in particular $N_G(X)\cap W=\emptyset$, $N_G(X)\cap Y=\emptyset$, and $X\cap Z=\emptyset$.
      By (W2)--(W3), every literal vertex in $G$ has exactly one literal neighbor, namely its opposite literal. Hence, the literal vertices in $N_G(X)$ are exactly $\overline{X}:=\{\overline{x}:x\in X\}$, a set of at most $k$ vertices, and $N_G(X)=\overline{X}\cup\big(N_G(X)\cap\mathcal{C}\big)$.

      Since $N_G(X)\cap W=\emptyset$, each $\overline{x}$ has $\deg_G(\overline{x})<\alpha m$, so the opposite literals contribute 
      $$\sum_{\overline{x}\in \overline{X}}\deg_G(\overline{x}) < k\alpha m = \rho m.$$

      Moreover, since $N_{G}(X)\cap Y=\emptyset$, we have $X\cap N_G(Y)=\emptyset$, so every $v\in X$ lies in $\mathcal{L}\setminus N_G(Y)$; together with $v\notin Z$ (as $X\cap Z=\emptyset$) this gives $\sum_{C\in N_G(v)\cap\mathcal{C}}\deg_G(C)<\alpha m$. Hence the clause neighbors contribute
       $$\sum_{C \in N_G(X)\cap\mathcal{C}} \deg_{G}(C) \leq \sum_{v \in X} \sum_{C \in N_{G}(v)\cap\mathcal{C}} \deg_{G}(C) < k\alpha m = \rho m.$$

       Adding the two bounds, we obtain $\sum_{v\in N_G(X)}\deg_G(v) < 2\rho m = \rho\sum_{v\in V(G)}\deg_G(v)$, contradicting the fact that $N_{G}(X)$ is a $\rho$-cover.
  \end{claimproof}

  To handle the cases where the above is too large, we first bound the size of $Y$ and $W$.
  \begin{claim}\label{claim:sampling-two}
      $|Y| \leq \frac{2(2t+1)k}{\rho}$ and $|W|\le \frac{2k}{\rho}$.
  \end{claim}
  \begin{claimproof}[Proof of Claim~\ref{claim:sampling-two}]
      Since there are no edges between clause vertices, we have \[\sum_{C \in \mathcal{C}} \deg_{G}(C) \le |E(G)| = m.\]
      If $|Y| > \frac{2(2t+1) k}{\rho}$,
      then $\sum_{C \in Y} \deg_{G}(C) > \frac{2(2t+1)k}{\rho} \cdot \frac{\rho}{2(2t+1)k} \cdot m = m$, a contradiction.

        Furthermore, since $\sum_{\ell\in\mathcal{L}}\deg_G(\ell)\le \sum_{v\in V(G)}\deg_G(v)=2m$, at most $\frac{2m}{\alpha m}=\frac{2k}{\rho}$ literal vertices have degree at least $\alpha m$, so $|W|\le \frac{2k}{\rho}$.
  \end{claimproof}

  By the previous claim, if $W\cup Y\cup Z$ exceeds our size bound, it is due to $Z$ being very large. The following claims show that in this regime, the graph contains a biclique, which we can find efficiently, and that any such biclique yields a valid output set as well. The case $|\mathcal{C} \setminus Y| < 2t+2$ is handled separately below.

  \begin{claim}\label{claim:sampling-three}
      Suppose $|Z| > (2t+2)/\gamma$ and $|\mathcal{C} \setminus Y| \geq 2t+2$, and let $Z' \subseteq Z$ be an arbitrary subset of size $\lfloor(2t+2)/\gamma\rfloor +1$ (which exists since $|Z| > (2t+2)/\gamma$).
      Then $G[Z' \cup N_G(Z')]$ contains a $K_{2t+2,2t+2}$ (a complete bipartite graph on $4t+4$ vertices) as a subgraph, with one side contained in $Z'\subseteq\mathcal{L}$ and the other in $\mathcal{C}\setminus Y$.
  \end{claim}
  \begin{claimproof}[Proof of Claim~\ref{claim:sampling-three}]
      Recall that for every literal vertex $v\in Z$ we have $\sum_{C\in N_G(v)\cap\mathcal{C}}\deg_G(C)\ge \alpha m$ and, since $v\notin N_G(Y)$, each clause neighbor $C\in N_G(v)\cap\mathcal{C}\subseteq \mathcal{C}\setminus Y$ satisfies $\deg_G(C)<\beta m$.

      Sample a set $S$ of $2t+2$ distinct clause vertices from $\mathcal{C}\setminus Y$, drawn one at a time without replacement, each clause with probability proportional to its degree.

      Fix a literal $x \in Z'$ and consider the probability that all $2t+2$ sampled clauses are neighbors of $x$. For each $j\in\{0,1,\ldots,2t+1\}$, condition on the first $j$ draws being neighbors of $x$. These $j$ removed clauses lie in $\mathcal{C}\setminus Y$ and so have total degree less than $j\beta m$, hence the remaining neighbors of $x$ still carry mass at least $\alpha m-j\beta m$, while the normalizer---the total degree of the undrawn clauses of $\mathcal{C}\setminus Y$---is at most $\sum_{C\in\mathcal{C}}\deg_G(C)\le m$, as the incident edges account for at most $m$ of the edges. Thus the $(j+1)$-th draw is a neighbor of $x$ with probability at least $\frac{\alpha m-j\beta m}{m}=\alpha-j\beta$. Multiplying over $j=0,1,\ldots,2t+1$ gives
      $$\Pr[x\text{ is adjacent to every vertex of }S]\ge \alpha(\alpha-\beta)(\alpha-2\beta)\cdots(\alpha-(2t+1)\beta)=\gamma.$$

      By linearity of expectation, the expected number of literals in $Z'$ that are adjacent to every vertex of $S$ is at least $|Z'|\cdot \gamma > \frac{2t+2}{\gamma}\cdot \gamma=2t+2$. As the maximum is at least the expectation, there is a choice of~$S$ for which at least $2t+2$ literals of $Z'$ are adjacent to every clause of $S$. These $2t+2$ literals together with $S$ form a $K_{2t+2,2t+2}$ in $G[Z' \cup N_G(Z')]$.
  \end{claimproof}

  \begin{claim}\label{claim:sampling-four}
        If $|Z| > (2t+2)/\gamma$ and $|\mathcal{C} \setminus Y| \geq 2t+2$,
       then one can find a $K_{2t+2,2t+2}$ in $G$ with one side in $\mathcal{L}$ and the other in $\mathcal{C}$
       in time $(\frac{2k}{\rho})^{\mathcal{O}(t^2)} \cdot m$.
  \end{claim}
  \begin{claimproof}[Proof of Claim~\ref{claim:sampling-four}]
      Let $Z' \subseteq Z$ be an arbitrary subset of $Z$ of size $\lfloor(2t+2)/\gamma\rfloor +1$.
      From Claim~\ref{claim:sampling-three}, we know that $G[Z' \cup N_G(Z')]$ contains a $K_{2t+2,2t+2}$.
      Thus a $K_{2t+2,2t+2}$ can be found by brute-forcing over all choices of $2t+2$ vertices from $Z'$ in time $\mathcal{O}\big( \binom{|Z'|}{2t+2} \cdot (2t+2) \cdot m\big)$.
      Since $\gamma \ge (\alpha-(2t+1)\beta)^{2t+2} = (\frac{\rho}{2k})^{2t+2}$, we have
      $|Z'| \le (2t+2)(\frac{2k}{\rho})^{2t+2}+1
      =(\frac{2k}{\rho})^{\mathcal{O}(t)}$.
      Consequently,
      $\binom{|Z'|}{2t+2} \le |Z'|^{2t+2}
      =(\frac{2k}{\rho})^{\mathcal{O}(t^2)}$, which gives the claimed running time.
  \end{claimproof}

  \begin{claim}\label{claim:sampling-five}
      Let $(A,B)$ be a $K_{2t+2,2t+2}$ in $G$ with $A\subseteq \mathcal{L}$ and $B\subseteq \mathcal{C}$, and let $A':= A \cup \{\overline{a} : a\in A\}\subseteq \mathcal{L}$, where $\overline{a}$ denotes the opposite literal vertex of $a$.
      Then for any solution $X$ to $I$, we have that
      either $X \cap A' \neq \emptyset$
      or $N_G(X) \cap B \neq \emptyset$.
  \end{claim}
  \begin{claimproof}[Proof of Claim~\ref{claim:sampling-five}]
      Suppose, for the sake of contradiction, that there is a solution $X$ to $I$ such that $X \cap A' = \emptyset$ and $N_G(X) \cap B = \emptyset$. We show that the biclique $(A,B)$ survives in the residual graph $G-N_G[X]$.
      Each literal vertex $a\in A$ lies outside $N_G[X]$: otherwise $a\in X$, or $a$ is adjacent to some $x\in X$, which (as $a$ is a literal vertex) can only be via the opposite-literal edge, i.e.\ $\overline{a}\in X$; either case yields $X\cap A'\neq\emptyset$, a contradiction.
      Likewise each clause vertex $b\in B$ lies outside $N_G[X]$: being a clause vertex, $b\in N_G[X]$ would mean $b\in N_G(X)$, contradicting $N_G(X)\cap B=\emptyset$.
      Hence $G-N_G[X]$ contains the biclique $(A,B)$, whose treewidth is $2t+2$; thus $G-N_G[X]$ has treewidth at least $2t+2$. By Observation~\ref{obs:treewidth-translation}, however, $G-N_G[X]$ has treewidth at most $2t+1$ since $X$ is a solution to $I$, a contradiction.
  \end{claimproof}

  The algorithm outputs the set $V'$ based on the following cases. We prove the correctness of each individual case.

    \begin{enumerate}
        \item Case 1: $|\mathcal{C} \setminus Y| \leq 2t+1$.

          Here $Z=\emptyset$: any $v\in\mathcal{L}\setminus N_G(Y)$ has all its clause neighbors in $\mathcal{C}\setminus Y$, each of degree below $\beta m$, so $\sum_{C\in N_G(v)\cap\mathcal{C}}\deg_G(C) < (2t+1)\beta m = \tfrac{\alpha m}{2} < \alpha m$ and hence $v\notin Z$.
          Output $V' = W \cup Y$. By Claim~\ref{claim:sampling-one} (with $Z=\emptyset$), every solution $X$ satisfies $N_G[X]\cap(W\cup Y)\neq\emptyset$, so $N_G[X]\cap V'\neq\emptyset$ since $V'=W\cup Y$. By Claim~\ref{claim:sampling-two}, $|V'|\le|W|+|Y|=(\frac{2k}{\rho})^{\mathcal{O}(t)}$.

        \item Case 2: $|\mathcal{C} \setminus Y| \geq 2t+2$ and $|Z| \leq (2t+2)/\gamma$.

        Output $V' = Y\cup W\cup Z$. The correctness is immediate from Claim~\ref{claim:sampling-one}: every solution $X$ satisfies $N_G[X]\cap V'=N_G[X]\cap(W\cup Y\cup Z)\neq\emptyset$. The size bound follows from Claim~\ref{claim:sampling-two} and $|Z| \leq (2t+2)/\gamma \le (2t+2)(\frac{2k}{\rho})^{2t+2}$. This gives
        $|V'|=(\frac{2k}{\rho})^{\mathcal{O}(t)}$.

        \item Case 3: $|\mathcal{C} \setminus Y| \geq 2t+2$ and $|Z| \geq \lfloor(2t+2)/\gamma\rfloor +1$.

        From Claim~\ref{claim:sampling-four}, a $K_{2t+2,2t+2}$, say $(A,B)$ with $A\subseteq\mathcal{L}$ and $B\subseteq\mathcal{C}$, can be found in $G$ in $(\frac{2k}{\rho})^{\mathcal{O}(t^2)} \cdot m$ time.
        Output $V' =A' \cup B$, where $A'$ is as defined in Claim~\ref{claim:sampling-five}.
        By Claim~\ref{claim:sampling-five}, every solution $X$ satisfies $X\cap A'\neq\emptyset$ or $N_G(X)\cap B\neq\emptyset$, so $N_G[X]\cap V'\neq\emptyset$.
        Moreover, $|V'|\leq |A'|+|B|\leq 3(2t+2)$.

    \end{enumerate}

        It remains to bound the running time. The degrees and the sets $Y$ and $W$ can be computed in $\mathcal{O}(n+m)$ time. We then mark $N_G(Y)$ and, in one further scan of the adjacency lists, compute the sums $\sum_{C\in N_G(v)\cap\mathcal{C}}\deg_G(C)$ that determine $Z$, again in $\mathcal{O}(n+m)$ time. Cases 1 and 2 then require only linear time, while Case 3 takes
        $(\frac{2k}{\rho})^{\mathcal{O}(t^2)}\cdot m$ time by Claim~\ref{claim:sampling-four}. This establishes the overall running time stated in the lemma.
\end{proof}

\subsection{The Progressor}
The progressor is the last module that gets called in the pipeline of Theorem~\ref{thm:mainTechnicalTheorem}. On a high level, this module receives an instance of \bdextproblemtw together with a vertex and creates a set of instances, such that in each, some progress towards finding a backdoor is made. As will be made precise below, progress can be achieved by decreasing the number of literals that can still be selected into the backdoor or by decreasing the $SL$-pattern size $s$ of the instance.
\progressorLemma*
\begin{proof}
    Recall that in $G$, every vertex is either a literal vertex or a clause vertex.

    \begin{claim}\label{cl:literal}
        The statement holds if $v$ is a literal vertex.
    \end{claim}
    \begin{claimproof}[Proof of Claim~\ref{cl:literal}]
        The construction for a literal vertex $v$ is guided by a hypothetical solution $X$ with $v\in N_G[X]$; the claimed properties (in particular the first item) are established unconditionally.
    Since~$v$ is a literal vertex, its only literal neighbor is $\bar v$; hence $v\in N_G[X]$ means that either $v\in X$ or $\bar v\in X$. In other words, we found a variable that belongs to our backdoor set, but we do not know yet, in which polarity we should set it.
    Consequently, we branch over the two possibilities ($v\in X$ or $\bar v\in X$) and over the index $i$ of the candidate set $P_i$ 
    out of which $u$ is selected.
    More specifically, for each $i\in[k]$ and each $u\in\{v,\bar v\}$ with $u\in P_i$, we create an instance $I_{i,u}=(k-1,s,t,G_{i,u})$ covering the case $\ell_i=u$.
    
    The graph $G_{i,u}$ is obtained from $G$ as follows:
    \begin{itemize}
        \item delete every clause vertex in $N_G(u)$, i.e.\ the clauses satisfied by $u$;
        \item delete all edges from $\bar u$ to clause vertices (as $\bar u$ is now false);
        \item if some clause vertex thereby becomes isolated, we discard the branch: such a clause can never be dominated, so $I_{i,u}$ would be a no-instance. We likewise discard the branch if removing $\bar u$ empties some remaining candidate set, i.e.\ if $P_j=\{\bar u\}$ for some $j\neq i$, whose only choice $\ell_j=\bar u$ would conflict with $u$;
        \item otherwise, remove the color indicating membership in candidate set $P_i$ from all literal vertices (the $i$-th literal is now resolved);
        \item remove $\bar u$ from all remaining candidate sets by removing the respective colors from $\bar u$ (choosing $\bar u$ would conflict with our choice of $u$); and
        \item remove all clause vertices $C'$ for which $\hat P_j\subseteq N_{G_{i,u}}(C')$ for some $j\in[k]\setminus\{i\}$ (this clause is satisfied for every remaining choice of $\ell_j$).
    \end{itemize}
    Note that we keep the literals $u$ and $\bar u$ in the instance since $u$ can be reused in another candidate set.
    There are at most $2k$ branches, each constructed in $\mathcal{O}(k(n+m))$ time: deleting the clauses of $N_G(u)$, deleting $\bar u$'s incident edges, and checking whether a clause became isolated together take $\mathcal{O}(n+m)$ time, while the last step tests, separately for each of the at most $k$ sets $\hat P_j$ with $j\neq i$, whether $\hat P_j\subseteq N_{G_{i,u}}(C')$ for every clause $C'$; as in the clause case below, each such test runs in $\mathcal{O}(n+m)$ time by a single scan over the incidence lists of the literals in $\hat P_j$. Summing over all branches yields a total running time of $\mathcal{O}(k^2(n+m))$.

    \paragraph*{Well-formedness.}
    $G_{i,u}$ arises from $G$ by deleting some clause vertices and the edges of $\bar u$ to clause vertices, while keeping all literal vertices and all opposite-literal edges; in particular the opposite-literal edges still form a perfect matching, so (W1)--(W3) are inherited.
    Regarding~(W4), each retained set $\hat P_j$ ($j\neq i$) is non-empty since branches violating this were discarded, and the final deletion removes exactly the clauses $C'$ with $\hat P_j\subseteq N_{G_{i,u}}(C')$ for some $j\neq i$; hence every surviving clause $C'$ satisfies $\hat P_j\not\subseteq N_{G_{i,u}}(C')$ for all $j\neq i$. Any $\hat P_j$ containing the now-isolated vertex $u$ satisfies this trivially, as $u\notin N_{G_{i,u}}(C')$.
    Regarding~(W5), consider any proper semi-ladder of $G_{i,u}$ whose literal side lies in $\hat P_j$; its clause side then consists of surviving clause vertices. All adjacencies and non-adjacencies among these vertices agree in $G$ and $G_{i,u}$, since the only edges $G_{i,u}$ lacks relative to $G$ are those incident to $\bar u$ (not a literal of $\hat P_j$) or to a deleted clause vertex (not on the surviving clause side). Hence the same sequence is a proper semi-ladder of $G$ with literal side in $P_j\supseteq\hat P_j$, giving $s_j'\le s_j$; thus the $SL$-pattern size of $G_{i,u}$ is at most $\sum_{j\neq i}s_j\le s$.

\paragraph*{Preservation of solutions and efficient solution transfer.}
    First, suppose $I$ is a yes-instance witnessed by a solution $X=\{\ell_1,\dots,\ell_k\}$ with $v\in N_G[X]$.
    Since $v\in N_G[X]$, either $u:=v$ or $u:=\bar v$ equals some $\ell_i\in X$, so $u\in P_i$.
    For every $j\neq i$, we have $\ell_j\neq\bar u$ (as $X$ is non-conflicting and $\ell_i=u$), so $\ell_j\in P_j\setminus\{\bar u\}=\hat P_j$; in particular, no $\hat P_j$ is empty.
    Moreover, the deletion of edges incident to $\bar u$ isolates no clause: a clause whose only neighbor is $\bar u$ lies in $G-N_G[X]$ but has no neighbor there, so it could not be dominated---contradicting that $X$ is a solution.
    Hence the branch $I_{i,u}$ is created and not discarded.
    Let $X'$ be the tuple $X$ with its $i$-th entry removed (its remaining entries $\ell_j$, $j\neq i$, may again include $u$). We claim that $X'$ is a solution to $I_{i,u}$ if and only if $X$ is a solution to~$I$.
    Admissibility and non-conflicting property transfer: $\ell_j\in\hat P_j$ for every $j\neq i$, and since the only literal that could conflict with $u$ is $\bar u\notin\hat P_j$, the tuple $X'$ is non-conflicting exactly when $X$ is.
    For the residual, every clause deleted in constructing $G_{i,u}$ lies in $N_G[X]$: those in $N_G(u)$ because $u\in X$, and every $C'$ deleted in the last step because $\ell_j\in\hat P_j\subseteq N_{G_{i,u}}(C')$ gives $C'\in N_G(\ell_j)\subseteq N_G(X)$; likewise, the deleted incidence edges are precisely those incident to $\bar u\in N_G[X]$.
    Hence $G-N_G[X]$ and $G_{i,u}-N_{G_{i,u}}[X']$ coincide (once the coloring of $P_i$ is forgotten), except that the latter still contains the isolated matched pair $\{u,\bar u\}$ whenever $u\notin X'$. As an isolated pair contracts to a single isolated vertex, the contractions $(G-N_G[X])^c$ and $(G_{i,u}-N_{G_{i,u}}[X'])^c$ have the same treewidth, and the clause vertices in one are dominated by a non-conflicting literal set if and only if this holds in the other (the isolated $u$ and $\bar u$ dominate no clause and need not be dominated).
    Consequently, $X'$ is a solution to $I_{i,u}$, so $\mathcal{I}$ contains the yes-instance $I_{i,u}$.
    
    Conversely, the same equivalence yields the solution transfer: given any solution $X'$ to a created branch $I_{i,u}$, its entries satisfy $\ell_j\in\hat P_j\subseteq P_j$ for all $j\neq i$ and are non-conflicting with $u$ (as $\bar u\notin\hat P_j$), so extending $X'$ by $\ell_i:=u$ yields a solution to $I$, obtained in $\mathcal{O}(n)$ time.
    \end{claimproof}

    \begin{claim}\label{cl:clause}
        The statement holds if $v$ is a clause vertex.
    \end{claim}
    \begin{claimproof}[Proof of Claim~\ref{cl:clause}]
    The construction for a clause vertex $v=C$ is guided by a hypothetical solution~$X$ with $C\in N_G(X)$; the properties claimed below, however, are established unconditionally.
    For each $i\in [k]$, if $P_i\cap N_G(C) \neq \emptyset$, we consider an instance $I_i$ covering the case that for the $i$-th literal vertex~$\ell_i$ of $X$ we have that $\ell_i\in N_G(C)$. Intuitively, this case corresponds to the clause $C$ being satisfied by the literal $\ell_i$.
    We now construct $G_i$ from $G$ as follows.
    Update the coloring so that the candidate sets of~$G_i$ become $\hat P_1,\dots,\hat P_k$, where $\hat P_i = P_i \cap N_G(C)$ and $\hat P_j=P_j$ for all $j\in [k]\setminus\{i\}$.
    Moreover, we remove all clause vertices $C'$ (and their incident edges) for which $\hat P_j \subseteq N_G(C')$ for some $j\in [k]$; such a clause $C'$ is adjacent to every admissible choice of $\ell_j\in \hat P_j$ and is thus intuitively guaranteed to be satisfied by whatever literal we end up selecting from $\hat P_j$. Note that in particular, this includes $C$ itself. Towards bounding the running time of the above procedure, note that each $I_i$ can be constructed in $\mathcal{O}(n+m)$ time: by~(W4) applied to $I$, no clause vertex $C'$ satisfies $P_j=\hat P_j\subseteq N_G(C')$ for any $j\neq i$, so only the condition $\hat P_i\subseteq N_G(C')$ has to be tested, which can be checked for all clause vertices simultaneously in $\mathcal{O}(n+m)$ time.

    \paragraph*{Well-formedness.}
    Since $G_i$ is obtained from $G$ by deleting only clause vertices and their incident edges, all literal vertices and opposite-literal edges remain, so (W1)--(W3) are inherited from $G$.
    Regarding~(W4), each set is non-empty: $\hat P_i = P_i\cap N_G(C)\neq\emptyset$ by the branching case, and $\hat P_j = P_j\neq\emptyset$ for $j\neq i$. Moreover, the deletion step removes exactly the clause vertices with $\hat P_j\subseteq N_G(C')$ for some $j\in [k]$, so every surviving clause $C'$ satisfies $\hat P_j\not\subseteq N_G(C')$ for all $j$ (and $N_G(C')=N_{G_i}(C')$, as clause vertices neighbor only literal vertices and no literal vertices are deleted).
    Finally,~(W5) follows from the next paragraph, which bounds the $SL$-pattern size of $G_i$ by $s-1$.

    \paragraph*{Shrinkage of $SL$-pattern size.}
    Let $(s_1, \ldots, s_k)$ and $(s_1', \ldots, s_k')$ be the $SL$-patterns of $I$ and some~$I_i$, respectively.
    For each $j\in [k]$, since $G_i$ is an induced subgraph of $G$ and $\hat P_j\subseteq P_j$, we have that every proper semi-ladder of $G_i$ with literals in $\hat P_j$ is also one of $G$ with literals in $P_j$; hence $s_j' \le s_j$.
    We sharpen this to $s_i' \le s_i-1$.
    Fix some $\ell \in P_i \setminus N_G(C)$, which exists since $P_i\not\subseteq N_G(C)$ by~(W4), and let $a_1,\dots,a_{s_i'}, b_1,\dots,b_{s_i'}$ be a proper semi-ladder of height $s_i'$ in $G_i$ with all literals in $\hat P_i$.
    Appending $(\ell,C)$ as the final ordered pair yields the sequence
    \[
        a_1,\dots,a_{s_i'},\ell,\qquad b_1,\dots,b_{s_i'},C,
    \]
    which is a proper semi-ladder of height $s_i'+1$ in $G$ with all literals in $P_i$: its new horizontal pair $\{\ell,C\}$ is a non-edge since $\ell\notin N_G(C)$, the diagonal edges $\{a_r,C\}$ are present for all $r\in[s_i']$ because $a_r\in \hat P_i\subseteq N_G(C)$, and all remaining edge and non-edge conditions are inherited.
    Here $\ell$ and $C$ are genuinely new vertices: $\ell\notin\{a_1,\dots,a_{s_i'}\}$ as $\ell\notin \hat P_i$, and $C$ is not a vertex of $G_i$.
    Hence $s_i\ge s_i'+1$, that is, $s_i'\le s_i-1$.
    Summing over $j$, the $SL$-pattern size of $I_i$ is at most $\sum_{j\neq i} s_j + (s_i-1) \le s-1$.

    \paragraph*{Preservation of solutions and efficient solution transfer.}
     Assume that $I$ is a yes-instance via a solution $X$ with literal vertices $\ell_1,\dots,\ell_k$ and let $C\in N_G(X)$. 
    Then $\ell_i\in N_G(C)$ for some $i\in[k]$, so $\ell_i\in P_i\cap N_G(C)=\hat P_i$ while $\ell_j\in P_j=\hat P_j$ for all $j\neq i$.
    We claim that $X$ is a solution to $I_i$ as well.
    First, each literal $\ell_j$ lies in $\hat P_j$.
    Secondly, every clause vertex~$C'$ removed in the construction of $G_i$ satisfies $\hat P_j\subseteq N_G(C')$ for some $j\in[k]$, and as $\ell_j\in \hat P_j$ this gives $C'\in N_G(\ell_j)\subseteq N_G(X)$.
    Hence all deleted vertices lie in $N_G[X]$, so $G-N_G[X]$ and $G_i-N_{G_i}[X]$ coincide once the $P_j$-colorings are forgotten; in particular their contractions (of all opposite-literal edges) $(G-N_G[X])^c$ and $(G_{i}-N_{G_{i}}[X])^c$ have the same treewidth and the clause vertices in $G-N_G[X]$ are dominated by a non-conflicting set of literal vertices if and only if this is the case in $G_{i}-N_{G_{i}}[X]$.
    Thus, since $X$ is a solution to $I$, it is a solution to $I_i$ as well.

    Conversely, let $X$ be any solution to $I_i$. Its literals satisfy $\ell_i\in\hat P_i=P_i\cap N_G(C)$ and $\ell_j\in\hat P_j=P_j$ for $j\neq i$, so $X$ is admissible for $I$ and in particular $C\in N_G(X)$. As above, every clause deleted in constructing $G_i$ lies in $N_G[X]$, so $G-N_G[X]$ and $G_i-N_{G_i}[X]$ coincide, and the equivalences of the previous paragraph show that $X$ is a solution to $I$. The solution is unchanged, so the third property holds in $\mathcal{O}(n)$ time.
    \end{claimproof}
    Finally, Lemma~\ref{lem:branchingReduction} follows directly from Claim~\ref{cl:literal} and Claim~\ref{cl:clause}.
\end{proof}

\subsection{Deriving the Main Result}\label{sec:solving-sat}

We are now ready to establish our main recursion routine.

\mainTechnicalTheorem*
\begin{proof}
    We first describe the algorithm, and then analyze its running time and argue its correctness.

    We apply the reducer (Lemma~\ref{lem:reducer}) to $I$ to obtain an instance $I'$ of size $\mathcal{O}_{k,s,t}(n+m)$ that is ready for extraction. We then use the extractor (Lemma~\ref{lem:extractor}) on $I'$ to obtain a set~$V'$ of vertices.
    Then, for each~$v\in V'$, we call the progressor (Lemma~\ref{lem:branchingReduction}) on $I'$ and $v$, each call yielding a set of at most $2k$ instances. We output all instances obtained across all $v\in V'$.

    Each of the three steps takes $\mathcal{O}_{k,s,t}(n+m)$ time by the respective lemma, and hence so does the entire pipeline.
    We output $q\le 2k\cdot |V'|$ instances. Recall from the extractor that $|V'|$ is bounded from above by a function of $k,s,t$ only, so $q$ is bounded by a similar function.

    The reducer preserves the values of $k,s$, and $t$, while the progressor preserves $t$ and, in each output instance, either decrements $k$ or decrements $s$. Thus, every output instance is of the form $(k-1,s,t,G_i)$ or $(k,s-1,t,G_i)$, establishing the first point. Moreover, since $I$ is well-formed, the reducer returns a well-formed instance $I'$ (Lemma~\ref{lem:reducer}), so the extractor and progressor receive a well-formed instance, and every instance produced by the progressor is well-formed (Lemma~\ref{lem:branchingReduction}); hence all $q$ output instances are well-formed.

    For the second point, first suppose $I$ is a yes-instance. Then by Lemma~\ref{lem:reducer} $I'$ is a yes-instance as well, witnessed by some solution~$X$. Let $G'$ be the graph in $I'$. Since $I'$ is ready for extraction, it satisfies the condition of Lemma~\ref{lem:extractor}, so $N_{G'}[X]\cap V' \neq \emptyset$; fix some $v\in N_{G'}[X]\cap V'$. As $v\in N_{G'}[X]$, the call to the progressor (Lemma~\ref{lem:branchingReduction}) on $I'$ and $v$ produces a yes-instance, which is among our output. Conversely, if one of the output instances is a yes-instance, then the third point (shown below) turns its solution into a solution to $I$, so $I$ is a yes-instance as well. This proves the equivalence.

    For the third point, consider a solution $X$ to an output instance $I_i$. Lemma~\ref{lem:branchingReduction} lifts $X$ to a solution~$X'$ to $I'$ and  Lemma~\ref{lem:reducer} then lifts $X'$ to a solution to $I$, each in $\mathcal{O}_{k,s,t}(n)$ time.
\end{proof}

Using the above, we are almost ready to establish the algorithm for solving \textsc{SAT} via weak backdoors (Theorem~\ref{thm:main}). 
The last ingredient we need is that the base instances of {\bdextproblemtw} after exhaustive application of Theorem~\ref{thm:mainTechnicalTheorem} are easy to solve. 

\begin{lemma}\label{lem:nonConflictSelect}
    There is an algorithm that takes as input a well-formed instance $I=(k,s,t,G)$ of {\bdextproblemtw} such that $G$ has $n$ literal vertices and no clause vertices. In time $\mathcal{O}(k^{k+1}+kn)$ it outputs a solution $X$ to $I$ or verifies that none exist.
\end{lemma}
\begin{proof}
    Note that contracting all edges between opposing literals in $G$ already gives a graph of treewidth~$0$. Thus, in order to solve an instance that is as specified by the lemma statement, 
    it only remains to decide whether there are pair-wise non-conflicting literals $\ell_1, \ldots, \ell_k$ such that $\ell_i \in P_i$ for all $i\in [k]$.
    We can enumerate all candidate sets in $\mathcal{O}(kn)$ time.
    From each candidate set $P_i$ with $|P_i| > k$ drop all but the $k$ first entries.
    Then, consider all $\mathcal{O}(k^k)$ possibilities to select one literal from each set.
    For each possibility, we can test whether it is non-conflicting in $\mathcal{O}(k)$ time, and if so output this choice as the solution $X$. 
    We reject the instance if the test fails for all choices.

    Clearly, any output choice is indeed a valid solution. 
    For the other direction, suppose there is a solution $X$. Consider the set of branches where for all sets $P_i$ with size smaller than $k$, the same literals are chosen as in $X$. Among these, there is a branch that finds a solution: The literals for sets $P_i$ with $|P_i|< k$ are non-conflicting between each other since they appear in $X$. We then fix the choices for the sets $P_i$ with $|P_i| \ge k$ (now $|P_i| = k$) one after another, at each step restricting to the subset of branches consistent with the choices made so far. When fixing such a set $P_i$, the at most $k-1$ literals chosen so far are opposite to at most $k-1$ elements. Since for the reduced set $|P_i|=k$, there is at least one branch where $\ell_i$ does not conflict with any previously chosen literal. Iterating over all these sets leaves a branch whose selection is non-conflicting overall.
\end{proof}

\begin{lemma}\label{lem:k-is-zero}
    There is an algorithm that takes as input a well-formed instance $I=(0,s,t,G)$ of {\bdextproblemtw} such that $G$ has $n$ vertices and $m$ edges. In time $2^{\mathcal{O}(t^3)}\cdot(n+m)$ it outputs a solution $X$ to $I$ or verifies that none exist.
\end{lemma}
\begin{proof}
    Since $k=0$ the only possible solution is the empty set. 
    Let $G^c$ be the graph obtained from contracting all edges between pairs of literals in $G$.
    Note that $(0,s,t,G)$ is a yes-instance (via $X=\emptyset$) if and only if $G^c$ has treewidth at most $t$ and there is a set of non-conflicting literal vertices that dominate all clause vertices.
    In time $2^{\mathcal{O}(t^3)}\cdot(n+m)$ we can decide whether $G^c$ has treewidth at most $t$, e.g. using
    Bodlaender's algorithm~\cite{Bodlaender96}. If the treewidth of $G^c$ exceeds $t$, we reject the instance.
    Otherwise, deciding whether there is a dominating set of non-conflicting literal vertices corresponds exactly to solving a \textsc{SAT} instance with incidence treewidth at most $t$. This can be done in $2^{\mathcal{O}(t)}\cdot(n+m)$ time, e.g.  by the dynamic-programming algorithm of Samer and Szeider~\cite{SamerS10}.
\end{proof}

We now proceed with the proof of Theorem~\ref{thm:main}. For convenience, we restate the theorem below.

\findingBackdoor*

\begin{proof}
    The algorithm proceeds in three steps: it translates $\varphi$ into an equivalent instance of \bdextproblemtw, iterates Theorem~\ref{thm:mainTechnicalTheorem} to reduce $k$ or the $SL$-pattern size to zero, and finally solves the resulting simple instances and backtracks a solution. Below, we describe and analyze each step individually.

 \paragraph*{Translation into \bdextproblemtw.}
    If $d=0$, then no formula is signed $d$-semi-ladder-free (every graph contains the empty semi-ladder of height $0$), so the statement holds vacuously; we therefore assume $d\ge 1$ in the following, which ensures $s=(d-1)k\ge 0$.
    Without loss of generality, assume that for $\varphi=(\var,\clause)$ there is no clause $C\in \clause$ with $L(\var)\subseteq C$, that is, no clause in $\varphi$ contains every variable of $\varphi$ in both polarities.
    Let $G_\varphi^*$ be the signed incidence graph of $\varphi$ but with the following coloring for edges and vertices: vertices are colored as \emph{literal} or \emph{clause} vertices depending on which they represent, respectively. Edges are colored according to $(W2)$ in the definition of \bdextproblemtw. 
    Every literal vertex also carries all $k$ additional colors---consequently the candidate sets $P_1, \ldots, P_k$ each contain all literal vertices.
    We claim that $\varphi$ has a weak $(k,t)$-backdoor $X$ if and only if $X$, read as a set of literal vertices of $G_\varphi^*$, is a solution to the instance $I_{dk}=(k,s=(d-1)k,t,G_\varphi^*)$ of \bdextproblemtw.
    Under the identification of each literal of $\varphi$ with the literal vertex of $G_\varphi^*$ representing it, the two notions match property by property:
    \begin{itemize}
        \item Because every $P_i$ contains all literal vertices, the constraint $\ell_i\in P_i$ is vacuous, so a solution is simply a set $X$ of at most $k$ literal vertices meeting the three properties below (note that the same literal can be chosen for multiple candidate sets $P_i$).
        \item The \emph{non-conflicting} property---no opposite-literal edge joins two chosen vertices---holds if and only if $X$ does not contain a variable in both polarities, i.e.\ if and only if the literals of $X$ define a consistent partial assignment $\alpha$ setting each of them to true.
        \item Deleting $N_{G_\varphi^*}[X]$ from $G_\varphi^*$ removes the variables underlying $X$ together with all clauses satisfied by~$\alpha$, and contracting the surviving opposite-literal edges merges the two vertices of each remaining variable into a single variable vertex; the resulting graph is exactly the signed incidence graph of $\varphi[\alpha]$. Hence the \emph{bounded-treewidth} property holds if and only if the signed incidence graph of $\varphi[\alpha]$ has treewidth at most $t$.
        \item A set of pairwise non-conflicting literal vertices of $G_\varphi^*-N_{G_\varphi^*}[X]$ that dominates all remaining clause vertices is precisely a consistent assignment of the residual variables that satisfies every remaining clause. Hence the \emph{domination} property holds if and only if $\varphi[\alpha]$ is satisfiable.
    \end{itemize}
    Combining these equivalences, $X$ is a solution to $I_{dk}$ if and only if $\alpha$ witnesses that $X$ is a weak $(k,t)$-backdoor of $\varphi$, proving the claim. Here the subscript represents $s+k = dk$, a measure of progress towards a simple instance. Moreover, $I_{dk}$ is well-formed.
    
    This translation takes $\mathcal{O}_k(|\var| + |\clause| + |\varphi|)$ time. For the following, let $n$ and $m$ denote the number of vertices and edges in $G_\varphi^*$, respectively, and note that $n = |\clause| + 2|\var|$ and $m = |\varphi| + |\var|$, so both are linear in the encoding size of $\varphi$.

 \paragraph*{Translation into simple instances.}
    Towards solving $I_{dk}$, we perform multiple iterations of Theorem~\ref{thm:mainTechnicalTheorem}, each time decrementing $k$ or the bound $s$ on the SL-pattern size, until we obtain a set of instances where one of the measures is $0$. If this is already the case for $I_{dk}$ (that is, $k=0$, or $s=0$ as happens when $d=1$), we skip this step and set $\mathcal{I}_1=\{I_{dk}\}$.
    Otherwise, let $\mathcal{I}_{dk}=\{I_{dk}\}$.
    For $p = dk-1, dk-2, \ldots, 1$, compute $\mathcal{I}_{p}$ by applying Theorem~\ref{thm:mainTechnicalTheorem} to every instance in $\mathcal{I}_{p+1}$ and collecting the union of all output sets in $\mathcal{I}_{p}$.

    Note that for all $p\in [dk]$, each instance $I'$ in $\mathcal{I}_p$ is of the form $(k',s',t,G_{I'})$ with $k'+s' = p$ for some graph $G_{I'}$, $k' \le k$, and $s' \le s$.
    If at some point we produce an instance with $k'=0$ or without clause vertices, we immediately put it in the set $\mathcal{I}_1$ instead of adding it to the current set.
    
    Indeed, we observe that all instances in $\mathcal{I}_1$ have $k'=0$ or do not have clause vertices. Suppose otherwise, i.e., there is an instance $I= (k', s', t, G) \in \mathcal{I}_1$ with $k'>0$ such that $G$ contains at least one clause vertex $C$.
    Since $k'>0$ and due to the existence of $C$, $I$ has not been placed in $\mathcal{I}_1$ prematurely but rather after $dk-1$ iterations of Theorem~\ref{thm:mainTechnicalTheorem}.
    Since initially $p=dk$ and each iteration decreases $p = s'+k'$, we have that $s'+k' = 1$ after the last iteration. As $k'>0$, we have $s' = 0$.
    However, we also have $P_1 \not\subseteq N_G(C)$ by~(W4).
    Consequently, there is a literal vertex $\ell \in P_1$ such that $\ell \notin N_G(C)$, so~$\ell$ and $C$ form a semi-ladder of height $1$ (i.e., a non-edge) in $G$.
    Thereby, $(k', s', t, G)$ has SL-pattern size at least $1$, contradicting $s'=0$.

     By Theorem~\ref{thm:mainTechnicalTheorem}, there is a constant $q$ depending only on $k,s,$ and $t$ such that $|\mathcal{I}_p| \le q \cdot |\mathcal{I}_{p+1}|$. Hence we construct $\mathcal{O}(q^{dk})=\mathcal{O}_{d,k,t}(1)$ many instances in total (recall that at every point $s' \le s < dk$). Each application of Theorem~\ref{thm:mainTechnicalTheorem} (and test for clause vertices) takes $\mathcal{O}_{d,k,t}(n+m)$ time, so the total time required for this step is in $\mathcal{O}_{d,k,t}(n+m)$.

 \paragraph*{Solving the simple instances and backtracking the solution.}
    All instances $(k', s', t, G)$ in $\mathcal{I}_1$ for which we have $k'=0$ are solved via Lemma~\ref{lem:k-is-zero}. All remaining instances in $\mathcal{I}_1$ are solved via Lemma~\ref{lem:nonConflictSelect} since, by the argument above, they do not have clause vertices.
    
    If there are only no-instances in $\mathcal{I}_1$, we certify that $\varphi$ does not admit a weak $(k,t)$-backdoor as follows: suppose $\varphi$ would admit a weak $(k,t)$-backdoor, then $I_{dk}$ is a yes-instance of \bdextproblemtw.
    Then, every set $\mathcal{I}_{p}$ for $p\in [dk]$ contains at least one yes-instance (via induction starting at $\mathcal{I}_{dk}$ and by Theorem~\ref{thm:mainTechnicalTheorem}) or at some point, a yes-instance was supposed to be placed in a set $\mathcal{I}_{p}$ with $p>1$ but instead was placed in $\mathcal{I}_1$ due to not having clause vertices or $k'=0$.
    In either case $\mathcal{I}_1$ would contain a yes-instance.
    
    Now suppose one of the instances in $\mathcal{I}_1$, say $I$, admits a solution $X$. Note that $I$ was obtained from~$I_{dk}$ by a sequence of applications of Theorem~\ref{thm:mainTechnicalTheorem}. By reversing this sequence, Theorem~\ref{thm:mainTechnicalTheorem} allows us to lift $X$ back up to a solution $X^\star$ to $I_{dk}$, which in turn is a weak $(k,t)$-backdoor to $\varphi$.

    To obtain a corresponding satisfying assignment, it suffices to solve the satisfiability problem on the residual CNF formula $\varphi^\star$ in which all variables represented in $X^\star$ and all clauses satisfied by $X^\star$ are removed. Since the incidence graph of this residual formula has treewidth at most $t$, any algorithm for solving \textsc{SAT} on instances of bounded incidence treewidth can be used, e.g.\ the dynamic-programming algorithm of Samer and Szeider~\cite{SamerS10}, which runs in time $2^{\mathcal{O}(t)}\cdot(n+m)$ once a width-$t$ tree decomposition is computed (for example via Bodlaender's algorithm~\cite{Bodlaender96} in time $2^{\mathcal{O}(t^3)}\cdot(n+m)$).

    Solving all instances in $\mathcal{I}_1$ takes $|\mathcal{I}_1|\cdot \mathcal{O}_{k,t}(n + m) = \mathcal{O}_{d,k,t}(n+m)$ time, using Lemma~\ref{lem:k-is-zero} for the instances with $k'=0$ and Lemma~\ref{lem:nonConflictSelect} for the clause-less instances. Lifting a solution through the iterations of Theorem~\ref{thm:mainTechnicalTheorem} takes $\mathcal{O}_{d,k,t}((dk-1) \cdot n) =
    \mathcal{O}_{d,k,t}(n)$ time.
\end{proof}

\section{Hardness of Finding Weak Backdoors}\label{sec:lowerbounds}

We accompany our algorithm by two lower bounds, which support our choice of measuring the height of semi-ladders in the signed incidence graph.
Under well-established complexity-theoretic assumptions, they rule out that the same result can be achieved when measuring ladders instead of semi-ladders in the signed incidence graph, or when measuring semi-ladders in the unsigned incidence graph. 
Both lower bounds rely on reductions from the \textsc{Red-Blue Dominating Set} problem.

\paramproblemdef{\textsc{Red-Blue Dominating Set}}
{ A graph $G=((R \uplus B), E)$ and a natural number $k$.}{$k$}
{Is there an $S \subseteq R$ with $|S|\leq k$ such that $S$ dominates all vertices in $B$?}

\begin{theorem}\label{thm:lowerbounds-ladderfree}
    For all $t \ge 0$ and $d\geq 18$, deciding whether a \textsc{SAT} instance $\varphi$ with an $L_d$-free signed vertex incidence graph admits a weak backdoor of size $k$ to treewidth $t$ is \W[1]-hard parameterized by $k$.
\end{theorem}
\begin{proof}
    It has been recently proved that $\textsc{Dominating Set}$ is \W[1]-hard on $L_d$-free instances, for $d \geq 17$~\cite[Theorem~5]{DBLP:journals/corr/abs-2602-07606}. To show this, the authors carefully inspect the standard parameterized reduction from $\textsc{Multicolored Independent Set}$ to $\textsc{Dominating Set}$ and argue that it gives rise to ladders of height at most 16. Their reduction is easily adapted to show that the problem $\textsc{Red-Blue Dominating Set}$ is \W[1]-hard on $L_d$-free instances, for $d \geq 17$ as well: we color all original vertices red and all vertices that the reduction introduces blue. The correctness of this adaptation follows from the fact that our budget only allows for choosing original vertices to dominate others anyway.

    We conclude our proof by reducing $\textsc{Red-Blue Dominating Set}$ on $L_d$-free instances, for any $t \ge 0$ and $d \geq 17$, to the problem of finding weak backdoors to treewidth $t$; as the bounded-ladder-height argument below shows, the signed incidence graph of the resulting formula is $L_{d+1}$-free, which yields the theorem for all $d \geq 18$. In the following, we will assume that all instances of $\textsc{Red-Blue Dominating Set}$ are pre-processed in a way that removes all vertices that are not adjacent to the respective other color: if there is a red vertex that is not adjacent to any blue vertex, it can be removed from the instance since it cannot cover any blue vertex; if there is a blue vertex that is not adjacent to any red vertex, then it witnesses a no-instance.
    Now, for a given $L_d$-free instance $((R \uplus B), E,k)$ of $\textsc{Red-Blue Dominating Set}$, we construct a formula $\varphi$ as follows.
    For every red vertex $r$, we create $t+1$ Boolean variables and let $\ell^r_1, \dots, \ell^r_{t+1}$ denote their positive literals; for every blue vertex $b$, we create $t+1$ clauses $C^b_1, \dots, C^b_{t+1}$. We distribute the literals over the clauses such that, for every $r \in R$, $b \in B$, and all $i,j \in [t+1]$, the literal $\ell^r_i$ occurs in the clause $C^b_j$ if and only if $\{r,b\}\in E$. In particular, all occurrences are positive, and $\{\ell^r_1, \dots, \ell^r_{t+1}\} \subseteq C^b_j$ holds for every $j \in [t+1]$ exactly when $\{r,b\}\in E$.

    \paragraph*{Bounded ladder height.}
    We show that any ladder in the signed incidence graph of $\varphi$ projects to a ladder of height smaller by at most one in the \textsc{Red-Blue Dominating Set} instance. Consequently, if the original instance is $L_d$-free for some $d \ge 17$, then the signed incidence graph of $\varphi$ is $L_{d+1}$-free.

    Recall that the signed incidence graph is almost bipartite between literals and clauses, with each literal joined to every clause containing it, together with the edges $\{x, \bar{x}\}$ between the two literals of each variable. In our construction every occurrence is positive, so a clause vertex is adjacent only to positive literals, a positive literal $\ell^r_i$ is adjacent only to the clauses containing it and to its negation $\bar{\ell}^r_i$, and each negation $\bar{\ell}^r_i$ has degree exactly one, its only neighbor being $\ell^r_i$.

    Consider any ladder in the signed incidence graph, with sides $a_1, \dots, a_h$ and $b_1, \dots, b_h$, where $\{a_i, b_j\}$ is an edge if and only if $i \le j$.

    We first claim that there is at most one index $i \in [h]$ for which $a_i$ and $b_i$ both represent literals. Suppose not, and let $i < j$ be two such indices. An edge between two literals must be one of the edges $\{x, \bar{x}\}$, so $\{a_i, b_i\}$ forces $b_i = \overline{a_i}$. By the ladder property, $\{a_i, b_j\}$ is an edge as well; since $a_i$ and $b_j$ are both literals, it is again of the form $\{x, \bar{x}\}$, whence $b_j = \overline{a_i} = b_i$. This contradicts $b_i \neq b_j$ and proves the claim.

    Delete this (at most one) exceptional position from the ladder and reindex the remaining ones; deleting a position from a ladder again yields a ladder, so we obtain one of height $h' \ge h - 1$ in which every pair $\{a_i, b_i\}$ is an edge between a literal and a clause. We show that one side of this sub-ladder consists entirely of positive literals and the other entirely of clauses. Up to the symmetry that exchanges the two sides (replacing $a_i, b_i$ by $b_{h'+1-i}, a_{h'+1-i}$), we may assume $a_1$ is a literal and $b_1$ a clause; being adjacent to a clause, $a_1$ is a positive literal. Its neighbors are clauses together with its negation $\bar{a}_1$, which has degree one and hence is adjacent to at most one ladder vertex; but the only $b$-side vertex adjacent to a single $a$-side vertex is $b_1$, which is a clause. Thus none of $b_1, \dots, b_{h'}$ equals $\bar{a}_1$, and since $a_1$ is adjacent to all of them, every $b_i$ is a clause. Each $a_i$ is then adjacent to the clause $b_{h'}$ and is therefore a positive literal.

    Finally, two vertices on the same side of a ladder have distinct neighborhoods towards the other side. As the literals $\ell^r_1, \dots, \ell^r_{t+1}$ of a fixed red vertex $r$ all have the same neighborhood in the incidence graph of $\varphi$, at most one of them occurs among the literals of the sub-ladder; symmetrically, at most one of the clauses $C^b_1, \dots, C^b_{t+1}$ of a fixed blue vertex $b$ occurs. Since $\ell^r_i$ and $C^b_j$ are adjacent if and only if $\{r, b\} \in E$, mapping each literal to its red vertex and each clause to its blue vertex yields a ladder of height $h'$ in the original instance. As $h' \ge h-1$, every ladder of height $h$ in the signed incidence graph projects to a ladder of height at least $h-1$ in the \textsc{Red-Blue Dominating Set} instance, which establishes the claimed $L_{d+1}$-freeness.

    \paragraph*{Correctness.} First assume that in an $L_d$-free instance of $\textsc{Red-Blue Dominating Set}$ there are $k$ red vertices $r_1, \dots, r_k$ that suffice to cover all blue vertices. Then, setting the variable underlying $\ell^{r_i}_1$ to true for each $i\in [k]$ satisfies all clauses of $\varphi$, hence the remaining instance (after contracting edges between opposing literals, thereby corresponding to the unsigned incidence graph) has treewidth zero. Conversely, assume that there is an assignment $\alpha$ to $k$ variables in $\varphi$ such that the contracted incidence graph of the remaining instance $\varphi[\alpha]$ has treewidth at most $t$. Observe that for every $b \in B$, the clauses $C^b_1, \dots, C^b_{t+1}$ form a biclique with the literals $\ell^r_1, \dots, \ell^r_{t+1}$ of at least one $r \in R$---a red neighbor of $b$, which exists by the preprocessing. This biclique is an induced $K_{t+1,t+1}$ and hence has treewidth $t+1$; since $\varphi[\alpha]$ has treewidth at most $t < t+1$, at least one vertex of the biclique must be absent from $\varphi[\alpha]$. That vertex is either a literal $\ell^r_i$ assigned by $\alpha$, or a clause $C^b_j$ satisfied by $\alpha$; in the latter case the satisfying literal is some $\ell^{r'}_i$ with $\{r',b\}\in E$ set to true. Either way, $\alpha$ assigns the variable underlying some $\ell^{r}_i$ with $\{r,b\}\in E$. Thus, the at most $k$ red vertices corresponding to the assigned variables dominate all blue vertices in the original instance.
\end{proof}

\begin{theorem}\label{thm:lowerbounds-unsigned}
    For all $t \ge 0$ and $d\geq 4$, deciding whether a \textsc{SAT} instance $\varphi$ with an $SL_d$-free unsigned incidence graph admits a weak backdoor of size $k$ to treewidth $t$ is \W[2]-hard parameterized by $k$.
\end{theorem}
\begin{proof}
    We reduce from $\textsc{Red-Blue Dominating Set}$, which is \W[2]-hard parameterized by~$k$ via its equivalence to \textsc{Set Cover}~\cite[Theorem~13.28]{CyganFKLMPPS15}.
    Given an instance $((R \uplus B),E,k)$, we construct the following formula~$\varphi$. For every red vertex $r$, we create $k\cdot(t+1)$ variables $r_{i,j}$ with $i \in [k]$, $j \in [t+1]$. For every blue vertex $b$, we create $t+1$ clauses $C^b_1, \dots, C^b_{t+1}$. We distribute the variables among the clauses as follows: for all $r \in R$ and $b \in B$, we inspect their adjacency. If $\{r,b\}\in E$, then for each $i \in [k]$, $j \in [t+1]$, we add the positive literal $r_{i,j}$  to each of the clauses $C^b_1, \dots, C^b_{t+1}$. If $\{r,b\}\notin E$, then for each $i \in [k]$, $j \in [t+1]$, we add the negative $\overline{r_{i,j}}$  to each of the clauses $C^b_1, \dots, C^b_{t+1}$.
    Additionally, we add $k\cdot (t+1)$ special clauses $C_{i,j}$ such that for each $i \in [k]$ and $j \in [t+1]$, the clause $C_{i,j}$ contains the positive literals $r_{i,1}, \dots, r_{i,t+1}$ for all $r \in R$.

    \paragraph*{Bounded semi-ladder height.}
The unsigned incidence graph of $\varphi$ contains no semi-ladder of height at least $4$ as follows. Each clause vertex $C^b_j$ originating from a blue vertex $b$ is adjacent to \emph{every} variable vertex, since every variable occurs either positively or negatively. Each special clause vertex $C_{i,j}$ is adjacent to exactly the variable vertices $R_i = \{r_{i,j'} \mid r \in R,\ j' \in [t+1]\}$, note that $R_i\cap R_{i'} = \emptyset$ for $i\neq i'$.

Since restricting a semi-ladder to any four of its positions again yields a semi-ladder, it suffices to rule out semi-ladders of height exactly~$4$. Suppose, for contradiction, that $a_1, \dots, a_4, b_1, \dots, b_4$ is such a semi-ladder. Recall that the incidence graph is bipartite between variable and clause vertices. If $a_1$ were a clause vertex, then its neighbors $b_2, b_3, b_4$ would all be variable vertices, and passing to the mirrored semi-ladder $b_4, \dots, b_1, a_4, \dots, a_1$ would place a variable vertex in the first position; we may therefore assume that $a_1$ is a variable vertex.

As $a_1$ is a variable vertex adjacent to $b_2, b_3, b_4$, these three are clause vertices; in turn, $a_1, a_2, a_3$ are adjacent to the clause vertex $b_4$ and hence are variable vertices. Consider now $b_2$ and $b_3$, whose diagonal partners $a_2$ and $a_3$ are variable vertices with $\{a_2, b_2\} \notin E$ and $\{a_3, b_3\} \notin E$. As blue clause vertices are adjacent to all variable vertices, both $b_2$ and $b_3$ must be special clause vertices, say with first indices $i_2$ and $i_3$. Then, by its adjacency, $a_1 \in R_{i_2} \cap R_{i_3}$. As the sets $R_i$ are pairwise disjoint, this forces $i_2 = i_3$ and thus $R_{i_2} = R_{i_3}$. But $a_2$ is adjacent to $b_3$ and not to $b_2$, so $a_2 \in R_{i_3} \setminus R_{i_2} = \emptyset$, a contradiction.

\paragraph*{Correctness.}
    First suppose that for a given instance of $\textsc{Red-Blue Dominating Set}$ there are at most $k$ red vertices $r^{(1)}, \dots, r^{(k)}$ that together cover all blue vertices. 
    Then assigning ${r^{(i)}}_{i,1}$ to true for each $i\in [k]$ will satisfy all clauses in $\varphi$ that originated from blue vertices. Furthermore, each special clause $C_{i,j}$ is satisfied by ${r^{(i)}}_{i,1}$. Thus, since all clauses are satisfied, the remaining instance has treewidth zero.

    For the converse direction, assume that there is an assignment $\alpha$ that assigns $k$ variables $S$ in $\varphi$ such that the 
    incidence graph of $\varphi[\alpha]$ has treewidth at most $t$.
    First, observe that for each $i \in [k]$, the clause vertices $C_{i, 1}, \dots, C_{i, t+1}$ form a
    biclique with $R_i$ and have no further neighbors.
    By assumption, $\varphi[\alpha]$ has treewidth at most $t$, which implies that for all $i\in [k]$, either (1) at least $|R|\cdot (t+1)-t$ variables from~$R_i$ need to be assigned by $\alpha$ or (2) at least one of the variables in $R_i$ must have been assigned by $\alpha$ in such a way that it satisfies one of the clauses $C_{i, 1}, \dots, C_{i, t+1}$. Without loss of generality we can assume $|R|> k$ (otherwise the original instance of $\textsc{Red-Blue Dominating Set}$ would have been trivial). Then option~(1) would require assigning at least $|R|\cdot(t+1)-t \ge |R| > k$ variables for a single~$i$ alone, exceeding the budget; hence option~(2) must be the case for all $i \in [k]$. Moreover, since for all $i \in [k]$, the clauses $C_{i, 1}, \dots, C_{i, t+1}$ contain only positive literals, we can strengthen the previous insight by remarking that for each $i \in [k]$, one of the variables in $R_i$ must have been assigned \emph{positively} by~$\alpha$, which---recalling the pairwise disjointness of the $R_i$---already requires us to set $k$ variables, thereby exhausting our entire budget; in particular, $\alpha$ assigns no variable negatively.
    
    Analogously, note that for each $b \in B$, the clause vertices $C_{1}^b, \dots, C_{t+1}^b$ form a biclique 
    with the set of all variable vertices, i.e. $\{r_{i,j} \mid i \in [k], j \in [t+1], r \in R\}$. Hence, for each $b \in B$, at least one of the clauses $C_{1}^b, \dots, C_{t+1}^b$ must be satisfied by $\alpha$. Since $\alpha$ assigns variables only positively, we can conclude that for any variable $r_{i,j}$ that satisfies $C_{1}^b, \dots, C_{t+1}^b$, we must have that $\{r, b \}\in E$ holds in the original instance of $\textsc{Red-Blue Dominating Set}$. Therefore, the set $\{r \in R \mid r_{i,j}\in S\}$, which has size at most~$k$, dominates all blue vertices.
\end{proof}

\section{Concluding Remarks}\label{sec:conclusion}

For all fixed integers $d,k,$ and $t$, we give a linear-time algorithm for detecting weak $t$-treewidth backdoors of size at most $k$ on \textsc{SAT} instances whose signed incidence graphs exclude semi-ladders of height $d$. Our result strictly generalizes previous algorithms
for bounded-clause-width \textsc{SAT}~\cite{FominLMRS15_uib} and for $K_{d,d}$-free \textsc{SAT}~\cite{DBLP:conf/icalp/LokshtanovP022} to a much richer family of dense instances.
By complementing our algorithmic contribution with matching lower bounds on (a) signed ladder-free and (b) unsigned semi-ladder-free instances, we demonstrate that signed semi-ladder-freeness is an
appropriate structural notion
for tackling dense \textsc{SAT} instances.

A natural direction for future work is to situate the tractability of detection of a weak $t$-treewidth backdoor within the landscape of dense graph classes obtained as transductions of sparse classes, as surveyed in~\cite{DBLP:journals/corr/abs-2501-04166} (see e.g.~\cite[Figure~1]{DBLP:journals/corr/abs-2501-04166}). We highlight two especially relevant dense structural notions here: {\em monadically stable} and {\em bounded twin-width} graph classes, both of which have seen considerable algorithmic success and mark independent tractability frontiers for FO model checking in the dense regime~\cite{DBLP:conf/focs/DreierEMMPT24,DBLP:journals/jacm/BonnetKTW22}.
 Monadically stable graph classes 
 are exactly those that do not transduce the class of all half-graphs (ladders). 
 One of our hardness results shows that just excluding large ladders is not enough for tractability, but
 excluding the presence of large ladders as transductions---which leads to monadically stable graph classes---may lead to tractability.  Similarly, our lower bounds do not rule out efficient weak backdoor detection on instances of bounded twin-width.

On a more technical note, we remark that our current algorithmic write-up expects a bound on the signed semi-ladder height of an instance to be provided as part of the input.
 Determining a tight upper bound for this height is non-trivial and, to the best of our knowledge, there is no such algorithm achieving this in linear time. 
 Nevertheless, our algorithm can be adjusted such that, given integers $d,k,t$ and a \textsc{SAT} instance $\varphi$, it correctly outputs one of the following: (a) a satisfying assignment, (b) a semi-ladder of height~$d$, or (c) a guarantee that there is no weak $(k,t)$ backdoor.
In particular, it may provide a witnessing semi-ladder for cases where the algorithm fails due to $d$ not being an upper bound on the semi-ladder height. \\

\textbf{Acknowledgments.} 
This research was funded in whole or in part by the Austrian Science Fund (FWF), projects 10.55776/COE12 and 10.55776/Y1329
and by the Deutsche Forschungsgemeinschaft (DFG, German Research Foundation), project 563448077. Roohani Sharma acknowledges the support of the Young Scientist Fellowship of the Institute for Basic Science (IBS-R029-Y8).\\

\textbf{Declaration on the use of AI tools.} The authors used Claude Opus 4.7 and 4.8 as well as Claude Fable 5 (Anthropic) and Codex with GPT-5.5 (OpenAI) in the following ways. (i) To assist in the proofreading and polishing of our proofs and prose. (ii) To assist in producing figures and visualizations in tikz. 
The authors assume responsibility for all content.

\bibliographystyle{plainnat}
\bibliography{references, references-others}

\appendix

\section{\W[1]-hardness of \textsc{Minimum Weight SAT} on  \texorpdfstring{$K_{2,2}$}{K2,2}-free instances}

\paramproblemdef{\textsc{Multicolored Clique}}
{ A graph $G$, a natural number $k$ and a partition of $V(G)$ into $k$ independent sets $V_1, \ldots, V_k$ called the color classes.}{$k$}
{Does there exist $u_i \in V_i$, for each $i \in [k]$, such that $G[\{u_1, \ldots, u_k\}]$ is a clique?}

\paramproblemdef{\textsc{Minimum Weight SAT}}
{A CNF-\textsc{SAT} formula $\varphi$ and a natural number $k$}{$k$}
{Does there exist a satisfying assignment of $\varphi$ that assigns at most $k$ variables the value $1$?}

\begin{lemma}\label{lem:min-wt-sat}
    \textsc{Minimum Weight SAT} is \W[1]-hard on $K_{2,2}$-free instances.
\end{lemma}
\begin{proof}
    We give a parameter-preserving reduction from \textsc{Multicolored Clique}.
    Towards this, let $(G=(V_1, \ldots, V_k),k)$ be an instance of \textsc{Multicolored Clique}, where $V_1, \ldots, V_k$ are the color classes. 

    We create a \textsc{SAT} formula, whose variable set is the same as the set of vertices of $G$, that is $\bigcup_{i=1}^k V_i$. The clause set is defined as follows: 
    for each $i \in \{1, \ldots,k\}$,
    create a clause containing a positive literal corresponding to each vertex in $V_i$. Such clauses are called \texttt{color-class-clauses}.
    Further for each $u \in V_i$ and $v \in V_j$, such that $\{u,v\} \not \in E(G)$,
    create a clause containing a negative literal of $u$ and a negative literal of $v$.
    Such clauses are called \texttt{non-adjacency-clauses}.
    Then $(\varphi,k)$ is an instance of \textsc{Minimum Weight SAT}, where $\varphi$ is the formula as created above.

    For the correctness, if $\{u_1, \ldots, u_k\}$ is a clique in $G$, where $i \in \{1, \ldots, k\}$,
    then consider an assignment of $\varphi$ where all variables corresponding to the vertices of 
    $\{u_1, \ldots, u_k\}$ are assigned $1$ and all other vertices are assigned $0$. This assignment satisfies all \texttt{color-class-clauses} because there is one vertex from each color class whose variable is assigned $1$. Secondly, for each non-adjacent pair of vertices of $G$, both the end-points are not chosen to be in the clique $\{u_1, \ldots, u_k\}$, such the variables corresponding to these end-points are not both set to $1$, thus the \texttt{non-adjacency-clauses} are satisfied.

    The other direction works similarly.
    Consider a satisfying assignment $\alpha$ for the formula $\varphi$ which has at most $k$ ones.
    We claim that the vertices of $G$ corresponding to the variables of $\varphi$ that $\alpha$ assigns to $1$, induce a clique of size $k$ in $G$, each of whose vertex belongs to a different color class.
    Indeed, since $\alpha$ satisfies the \texttt{color-class-clauses},
    it sets exactly one variable from each color class vertex $1$. Let $u_i,u_j$ be vertices corresponding to two variables that are set to $1$ by $\alpha$. Then $\{u_i,u_j\} \in E(G)$ because otherwise it would contradict a \texttt{non-adjacency-clauses}.
    
\end{proof}
\end{document}